\documentclass[11pt]{article}
\usepackage[T1]{fontenc}
\usepackage[margin=1in]{geometry}
\usepackage{tempora}

\usepackage[linesnumbered,ruled,vlined]{algorithm2e}
\usepackage{amsfonts}
\usepackage{amsmath}
\usepackage{amssymb}
\usepackage{amsthm}
\usepackage{authblk}
\usepackage{bbm}
\usepackage{enumitem}
\usepackage{float}
\usepackage[para]{footmisc}
\usepackage{graphicx}
\usepackage{mathtools}
\usepackage{physics}
\usepackage{thmtools}
\usepackage{titling}
\usepackage{xcolor}

\makeatletter
\long\def\@makefntext#1{\leavevmode
    \@makefnmark\nobreak
    #1%
}
\makeatother

\theoremstyle{plain}
\newtheorem{theorem}{Theorem}
\newtheorem{lemma}[theorem]{Lemma}
\newtheorem{proposition}[theorem]{Proposition}
\theoremstyle{definition}
\newtheorem{definition}[theorem]{Definition}

\definecolor{medblue}{RGB}{0, 0, 100}
\definecolor{panblue}{RGB}{0,15,125}
\definecolor{carmine}{RGB}{150, 0, 24}
\usepackage[
    colorlinks=true,
    linkcolor=carmine,
    anchorcolor=red,
    citecolor=medblue,
    urlcolor=panblue
]{hyperref}

\usepackage[noabbrev,capitalize]{cleveref}
\crefname{equation}{Eq.}{Eqs.}

\providecommand{\centcol}{\mathrel{\mathop{:}}}
\newcommand{\no}[1]{\,\centcol\mathrel{#1}\centcol\,}
\newcommand{\kb}[1]{\ketbra{#1}{#1}}
\newcommand{\ho}{\text{h.o.}}
\newcommand{\E}{\mathbb{E}}
\newcommand{\e}{\mathrm{e}}
\newcommand{\bigO}{\mathcal{O}}

\newcommand{\changepoint}{changepoint}
\newcommand{\Changepoint}{Changepoint}
\newcommand{\HamCert}{\textnormal{\textsc{HamCert}}}
\newcommand{\HamChangePoint}{\textnormal{\textsc{HamChangePoint}}}

\newcommand{\pos}[1]{(#1)^+}
\newcommand{\delay}[1]{\overline{\E}[#1]}
\newcommand{\falsealarm}[1]{\E_\infty[#1]}
\let\oldtextsc\textsc
\renewcommand{\textsc}[1]{{\fontfamily{cmr}\oldtextsc{#1}}}

\DeclareMathOperator*{\esssup}{ess\,sup}
\DeclareMathOperator{\Ber}{Ber}

\NewDocumentCommand{\afffunc}{ m }{%
    \iffirstitem%
        \firstitemfalse%
    \else
        ,%
    \fi%
    \ref{affil#1}%
}
\NewDocumentCommand{\aff}{ >{\SplitList{,}}m }{%
    \newif\iffirstitem%
    \firstitemtrue%
    $^\text{%
        \ProcessList{#1}{\afffunc}%
        \firstitemtrue%
    }$%
}

\title{Autonomous Hamiltonian certification and \changepoint{} detection}
\date{March~30, 2026}

\author{
    Steven~T.~Flammia\aff{1,2},\,
    Dmitrii Khitrin\aff{3},\,
    Muzhou~Ma\aff{4,5},\authorcr
    Jamie~Sikora\aff{1},\,
    Yu~Tong\aff{3,6,7},\, and
    Alice~Zheng\aff{1}
}

\makeatletter
\hypersetup{
    pdftitle={\@title},
    pdfauthor={
        Steven T. Flammia,
        Dmitrii Khitrin,
        Muzhou Ma,
        Jamie Sikora,
        Yu Tong,
        Alice Zheng
    }
}
\makeatother

\newcommand{\afftarget}[2]{\footnotetext[#1]{\phantomsection\label{affil#1}#2}}

\begin{document}
\maketitle

\afftarget{1}{Department of Computer Science, Virginia Tech}
\afftarget{2}{Phasecraft Inc.}
\afftarget{3}{Department of Electrical and Computer Engineering, Duke University}
\afftarget{4}{Institute for Quantum Information and Matter, Caltech}
\afftarget{5}{Department of Computing and Mathematical Sciences, Caltech}
\afftarget{6}{Duke Quantum Center, Duke University}
\afftarget{7}{Department of Mathematics, Duke University}
\vspace*{-1.5em}

\begin{abstract}
    Modern quantum devices require high-precision Hamiltonian dynamics, but environmental noise can cause calibrated Hamiltonian parameters to drift over time, necessitating expensive recalibration.
    Detecting when recalibration is needed is challenging, especially since the very gates required for sophisticated verification protocols may themselves be miscalibrated.
    While cloud quantum computing services implement heuristic routines for triggering recalibration, the fundamental limits of optimal recalibration are yet to be illuminated.
    Here, we study the recalibration problem by developing efficient Hamiltonian certification and \changepoint{} detection protocols in the \emph{autonomous} setting.
    In this setting, we cannot rely on an external noiseless device and additionally use only single-qubit gates and measurements, making the protocols robust to the calibration issues for multi-qubit operations they aim to detect.
    For unknown $n$-qubit Hamiltonians $H$ and $H_0$ with operator norm $\norm{H}, \norm{H_0}$ bounded by $M$, our certification protocol distinguishes whether $\norm{H - H_0}_F \geq \epsilon$ or $\norm{H - H_0}_F \leq O(\epsilon/\sqrt{n})$ with sample complexity $\bigO(nM^2\ln(1/\delta)/\epsilon^2)$ and total evolution time $\bigO(nM\ln(1/\delta)/\epsilon^2)$, where $\delta$ bounds the failure probability.
    The protocol achieves this by evolving random stabilizer product states and performing adaptive single-qubit measurements based on a classically simulable hypothesis state.
    Extending this to continuous monitoring, we develop an online \changepoint{} detection algorithm using the CUSUM procedure that achieves a detection delay time bounded by $\bigO(nM\ln(M\falsealarm{T})/\epsilon^2)$, matching the known asymptotically optimal scaling with respect to false alarm run time $\falsealarm{T}$.
    Our approach enables quantum devices to autonomously monitor their own calibration status without requiring ancillary systems, entangling operations, or a trusted reference device, offering a practical solution for robust quantum computing with contemporary noisy devices.
\end{abstract}

\section{Introduction}

Implementing high-precision Hamiltonian dynamics is essential for advancing both analog quantum simulation and digital quantum computation~\cite{Daley23,ProbingQuantumDynamics,SpinLiquid, PhysRevX.11.041058,demo_FTgates,LogicalRydberg,GoogleBreakEven,AlgorithmicFT,QuantumVerification}.
In practice, environmental noise causes calibrated parameters to drift over time, degrading performance and requiring frequent recalibration that incurs substantial downtime.

Detecting when recalibration is needed poses a significant challenge: because quantum computations target classically intractable problems, standard validation methods are impractical as the correct outputs are often unpredictable.
While quantum error correction is in principle a remedy, its effectiveness fundamentally depends on accurate gate implementations, which in turn require precise device calibration.
For the same reason, multi-qubit operations become increasingly unreliable in the presence of such drift.

We study the following problem: can a quantum device detect drifts in \emph{its own} Hamiltonian dynamics?
We call the resulting procedures \emph{autonomous} as they are not assisted by ancillas or external devices, and require them to only use single-qubit state preparation and measurements.

Recent years have seen a surge of proposals for Hamiltonian certification, property testing, and learning\cite{RobustLearningSuperConducting,HuangTongFangSu2023learning,ma2024learningkbodyhamiltonianscompressed,LiTongNiGefenYing2023heisenberg,QLFH,Bakshi_2024,PhysRevA.110.062421,francca2022efficient,Yu2023robustefficient,SampleEfficientLearning,EntanglementHamiltonianLearning,Olsacher_2025,2024arXiv240101308O,PauliTransfer,MobusBluhmCaroEtAl2023dissipation,Hu_2025}, which---despite being in principle applicable to detecting miscalibration---typically make assumptions unsuitable for the setting.
In particular, access to an ancillary system, the ability to perform complicated multi-qubit gates, or sophisticated quantum controls are inherently circular in the context of recalibration, as performing these operations in existing quantum experiment platforms already requires the device to be calibrated.
This chicken-and-egg problem underscores the need for a protocol that uses only experimentally feasible controls to detect miscalibration.

Quantum state learning and testing protocols also offer relevant techniques \cite{buadescu2019quantum, huang2025certifying, gupta2025single}, though the recalibration problem demands additional considerations beyond the limited access model.
Typical research in this area focuses on guarantees related to state fidelity or other distance notions between states.
However, when dealing with quantum dynamics we need to consider distance measures of Hamiltonians, and the recalibration setting reveals different metrics as natural bottleneck resources.
In particular, we aim to minimize the total evolution time and detection delay in order to identify Hamiltonian drifts without occupying the quantum device for too long.

To account for the unpredictable nature of miscalibrations, it is imperative to extend autonomous testing to the setting of continuous monitoring.
This aligns with a recent surge of proposals to apply various classical change detection methods in quantum scenarios \cite{akimoto_changepoint_2011,sentis_quantum_2016,sentis_exact_2017,sentis_online_2018,yu_experimentally_2018,mohan_generalized_2023,llorens2025quantum}.
Of particular relevance is the cumulative sum (CUSUM) procedure \cite{page1954continuous,lorden1971procedures,burr1994multivariate,lai_sequential_1995}, with its extensions being recently applied to detecting changes in sequences of quantum states \cite{fanizza_qusum_2022,john2025fundamental,xu2025quickest,gong_secure_2020,gong2025joint,gong2025quickest}.

More details about the problem we address can be seen in \cref{fig:main_fig}, below.
We envision the use of our procedures not only as standalone methods, but also in conjunction with and aiding any other quantum algorithms.
When executing a series of any experiments that require the dynamics of the device to be calibrated, data for the proposed change detection procedure can be periodically gathered, triggering recalibration when a change is detected.
The algorithms ensure that such a change is detected rapidly, and that contexts from previous unterminated certification blocks are carried over without the need for quantum memory.
In cases where a change is not detected within a single certification block, our change detection yields a maximum likelihood estimate for its location that can be used to discard and re-run any affected experiments.
\begin{figure}[ht]
    \centering
    \includegraphics[width=1.0\linewidth]{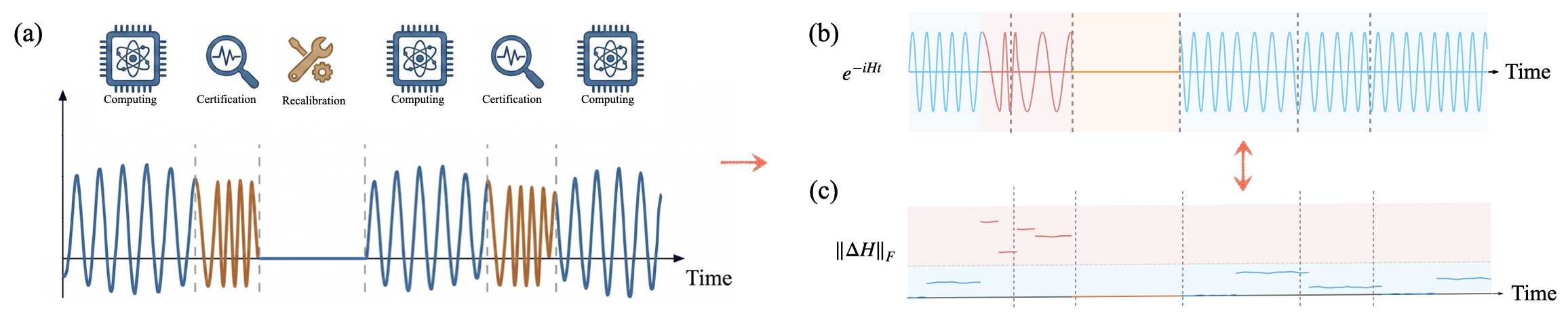}
    \caption{
        An example application of our \changepoint{} detection procedure to the recalibration problem.
        (a) The scheduling of a quantum device is divided into interleaved sections of certification (which run our change detection algorithm) and computing (which run other scheduled algorithms).
        The certification procedure can potentially trigger a recalibration procedure, which incurs a period of downtime.
        (b) Example Hamiltonian dynamics, with red regions indicating significant deviation from the target.
        The first shown certification section encounters this deviation and recalibrates; the second does not and proceeds as usual.
        (c) Same as the above in the frequency space.
        Lines within the red region indicate deviation above the upper threshold of our change detection procedure.
    }
    \label{fig:main_fig}
\end{figure}

\subsection{Main results}

As Hamiltonians in general can have an exponential number of terms, assumptions on their complexity are necessary to make relevant problems computationally tractable.
We do so by bounding the operator norm of each Hamiltonian by some $M$ and the magnitude of each term by a constant.
In particular, we let
\begin{equation}
    H = \sum_P \mu_P P,\quad
    P\in\{I,X,Y,Z\}^{\otimes n}\setminus\{I^{\otimes n}\},\quad
    \abs{\mu_P}\leq 1,\vspace*{-.5em}
\end{equation}
with $\norm{H} \leq M$.
The Hamiltonians are chosen to be traceless, since a two Hamiltonians whose difference is proportional to the identity have equivalent observable dynamics and differ only in an unobservable global phase.

The end goal of the setting we consider is to enable the creation of quantum devices robust to noise.
As the mere existence of such a noiseless device would defeat the purpose of considering the problem in the first place, we ensure that our algorithms do not have access to one.
That is, our proposed procedures are \emph{autonomous}, meaning they do not have access to an ancilla or any external quantum hardware.

Furthermore, to detect a drift in its own parameters, a noisy device needs to rely on operations that are not largely affected by this noise.
Specifically, we need to only use a particular subset of operations assumed to be robust---in this case, stabilizer product state preparation and arbitrary \emph{single-qubit} measurements.
This additionally implies we cannot reliably simulate the earlier mentioned evolution $\e^{-i H_0 t}$.

Our procedures are both single-qubit and autonomous, meaning they do not have:
\begin{enumerate}[label=\alph*)]
    \item Access to a second, noiseless device,
    \item An ancilla,
    \item Access to time evolution by $\e^{-i H_0 t}$ for an ideal Hamiltonian $H_0$ (either directly or via simulation),
    \item Access to entangled state preparation and measurements,
    \item Access to controlled evolution and inverse-time evolution,
    \item Access to Trotterization.
\end{enumerate}

At the core of our procedures is a single-sample method for testing Hamiltonians, which involves randomly preparing a state from an ensemble, evolving it by the unknown Hamiltonian, and measuring using a recently-proposed quantum state certification method \cite{gupta2025single}.
Our two main results---certification and \changepoint{} detection---differ in how they aggregate the classical statistics when this procedure is repeated, presenting the results either in terms of sample complexity or detection delay.

\paragraph{Certification:}

We provide an efficient Hamiltonian certification protocol assuming only limited product access to the quantum system.
The protocol time-evolves random stabilizer product states under the system Hamiltonian and performs adaptive single-qubit measurements.
This allows estimating the distance between the system and target Hamiltonians using the fidelity of the evolved and classically-simulated states, obtained with limited samples and total evolution time via an adaptively-chosen measurement basis.
Translating the fidelity between states to the Frobenius norm of the Hamiltonian difference is a delicate matter: if the evolution time is too short, then it is hard for the Hamiltonian difference to result in detectable infidelity; if the evolution time is too long, the perturbative approximation may break down and one cannot obtain a simple relation between the Hamiltonian difference and the fidelity.
Through careful analysis of the commutation relations of terms for time evolution, we identify an evolution time scaling that avoids the above problems.

We begin with a problem where the unknown $H$ does not change with time.
Our goal is to distinguish with fixed confidence between two hypotheses: $H$ and $H_0$ being near versus far.
\begin{definition}[Hamiltonian certification problem]\label{def:Ham_certification_prob}
    Given an unknown $n$-qubit lab Hamiltonian $H$ with bounded operator norm $M$, and a classically-known ideal $n$-qubit Hamiltonian $H_0$ with bounded operator norm $M$, the $\HamCert{}(\epsilon_1, \epsilon_2, \norm{\cdot})$ problem is to design a protocol that with high probability outputs
    \begin{itemize}[leftmargin=8em]
        \item[(Fail):] when $\norm{H-H_0} \geq \epsilon_2$,
        \item[(Pass):] when $\norm{H-H_0} \leq \epsilon_1$.
    \end{itemize}
\end{definition}
We solve this in a way that is efficient, both in the number of sampled (prepared and measured) states and the total time evolved with the given black-box device that implements $\e^{-i H t}$.
Below, we use $\norm{\cdot}_F$ to denote the normalized Frobenius norm.

\begin{theorem}[Adaptive Hamiltonian certification]\label{thm:certification}
    The $\HamCert(\bigO(\epsilon/\sqrt{n}), \epsilon, \norm{\cdot}_F)$ problem can be solved with probability at least $1 - \delta$ by an autonomous quantum algorithm using $N$ product state inputs and product basis measurements, and total evolution time $T$, where
    \begin{align}\label{eq:certification}
        N &= \bigO\left(\frac{nM^2\ln(1/\delta)}{\epsilon^2}\right),&
        T &= \bigO\left(\frac{nM\ln(1/\delta)}{\epsilon^2}\right).
    \end{align}
\end{theorem}

We remark that with an increasing system size $n$, environmental noise acts on more qubits and the normalized Frobenius norm would be expected to scale proportional to $\sqrt{n}$ in physically reasonable circumstances.
If this is accounted for by similarly scaling $\epsilon$ by $\sqrt{n}$, the above bounds on $N$ and $T$ are not dependent on $n$.

Additionally, the classical computation for our protocol is efficient as discussed in \cref{section:classical_certification_complexity}, scaling as a polynomial of the number of nonzero Pauli terms.
Corresponding sparsity assumptions are not required for the quantum protocol, with the shown results only needing the Hamiltonians to have bounded norm.
The same details apply to the subsequent \changepoint{} detection result in \cref{thm:online}, below.

\paragraph{\Changepoint{} detection:}

The above certification result provides a sample-efficient way of distinguishing the cases of $H$ and $H_0$ being close versus far.
In practice, the dynamics $H$ of a miscalibrated device need not be constant throughout time, instead following some time-dependent $H(t)$ that drifts away from $H_0$.
In the Hamiltonian \changepoint{} problem, we consider a discretization of this setting where $H(t)$ changes from $H_0$ to some $H_1$, then $H_2$ and so on.
Treating a time-dependent Hamiltonian as a sequence of time-independent ones in this manner is possible with relevant assumptions on the spectrum of the former.

While the certification procedure trivially extends to a sequence of $H_i$ for $i \in \{1,2,\ldots\}$ so long as the distance conditions in \cref{def:Ham_certification_prob} are satisfied for all $i$, in practice, quantum devices become more miscalibrated over time.
That is, $\norm{H_i-H_0}_F$ exceeds the larger threshold only after some \changepoint{} step $i \geq \nu$.
Our goal is to detect such a change as soon as possible, to initiate a recalibration procedure and avoid running the faulty device for a long time.
If on the other hand $\norm{H_i-H_0}_F$ stays relatively low for all $i$, we aim to prevent extraneous recalibration that incurs associated overhead.

\begin{definition}[Hamiltonian \changepoint{} problem]\label{def:ham_changepoint}
    Given a sequence of unknown $n$-qubit Hamiltonians $H_i$ for $i \in \{1, 2, \ldots\}$ and a classically-known ideal $n$-qubit Hamiltonian $H_0$, all having operator norm bounded by $M$, the $\HamChangePoint{}(\epsilon_1, \epsilon_2, \norm{\cdot})$ problem is to design a protocol that with high probability
    \begin{itemize}[leftmargin=8em]
        \item[(Terminates):] when $\norm{H_i - H_0} \geq \epsilon_2\ $ for all $i \in \{\nu+1, \nu+2, \ldots\}$ and any \changepoint{} location $\nu$,
        \item[(Keeps running):] when $\norm{H_i - H_0} \leq \epsilon_1\ $ for all $i \in \{1, 2, \ldots\}$.
    \end{itemize}
\end{definition}

In contrast to certification, which terminates after a fixed time with either hypothesis as the outcome, a change detection algorithm provides \emph{continuous inspection} by continuing execution in the case of a change not having happened.
To enable algorithms to terminate rapidly when a change does happen, we design them to be \emph{online} (i.e., \emph{streaming}), meaning they process data as soon as it becomes available and frequently output the current best guess for the hypothesis.

Attempting to apply the certification procedure to this problem directly leads to an issue where the accumulated statistics are influenced by the true \changepoint{} location $\nu$ in a way that is not possible to predict without knowing $\nu$.
While we could split the sequence into blocks of size $N$ (as in \cref{eq:certification}) that are certified separately, this incurs a delay overhead from termination decisions only being made at multiples of $N$.

We consider two objectives in this task: the delay elapsed on average between $\nu$ and when the change is detected, and the average run length until the algorithm falsely detects a change when there is none.
We use $\E_\nu$ and $\E_\infty$ to indicate conditioning on the two cases in \cref{def:ham_changepoint}, and $\overline{\E}$ to denote the largest $\E_\nu$ over finite $\nu$ and $H_1, \ldots, H_\nu$ (see \cref{def:delay}).
Letting $N$ be the round after which an algorithm terminates, the number of samples until false termination when Hamiltonians are $\epsilon_1$-close is represented by $\falsealarm{N}$, and the detection delay when they start (and keep) being $\epsilon_2$-far is $\delay{\pos{N - \nu}}$ (where $\pos{x} = \max(x, 0)$).
We additionally consider the same metrics in terms of the total evolution time $T$; for an algorithm that evolves using the Hamiltonians for the same time $t$ at each step, $\delay{\pos{T-t\nu}} = t \delay{\pos{N-\nu}}$ and $\E_\infty[T] = t \E_\infty[N]$.

It is known that an asymptotically-optimal procedure for the classical version of this task is given by an algorithm called CUSUM \cite{page1954continuous,lorden1971procedures}.
Applying it to an intermediate single-sample result used to obtain \cref{thm:certification}, we obtain the following performance bounds.

\begin{theorem}[Autonomous Hamiltonian change detection]\label{thm:online}
    The $\HamChangePoint{}(\bigO(\epsilon/\sqrt{n}), \epsilon, \norm{\cdot}_F)$ problem can be solved by an autonomous online quantum algorithm using a single product state input and a product basis measurement per step.
    The algorithm terminates in $N$ steps and total evolution time $T$, where
    \begin{align}\label{eq:cusum-adaptive}
        \delay{\pos{N-\nu}} &\leq \bigO\left( \frac{n M^2 \ln \falsealarm{N}}{\epsilon^2} \right),&
        \delay{\pos{T-t\nu}} &\leq \bigO\left( \frac{n M \ln (M \falsealarm{T})}{\epsilon^2} \right).
    \end{align}
\end{theorem}

The above relationship between detection delay and runtime till false alarm is known to be optimal in the \changepoint{} setting, with dependence on other parameters similar to that of \cref{thm:certification}.
In addition, the CUSUM-based algorithm yields a maximum-likelihood estimator for the \changepoint{} location when one is detected \cite{lorden1971procedures}, meaning that it allows a quantum device not only to monitor its own performance in a robust manner, but also to void and re-run any computations affected by a change after it is detected.

Our motivation for developing this \changepoint{} procedure stems from a particular practical problem.
Suppose one has access to a quantum device with robust realizations of single-qubit gates and measurements, yet with entangling gates that suffer from noise and accumulating miscalibration.
We wish to run algorithms on this device that require all gates to be properly functioning.
As illustrated in \cref{fig:main_fig}, we can interlace instances of algorithm execution with blocks of time allotted for miscalibration checks.
The introduced \changepoint{} procedure is particularly useful for this, as it can trigger recalibration and do so with several important properties.
The miscalibration checks are not affected by the miscalibrations themselves, so an affected quantum computer can run this check on itself.
Rapid change detection avoids running the device when it is malfunctioning, and since the method yields an estimator for the occurring \changepoint{} location, we also learn which algorithm runs' results are affected by miscalibration and can re-run them.
This presents a practical solution to the issues faced by contemporary quantum devices, ensuring robust computation.

\subsection{Comparison to alternative procedures}

Our proposed protocols achieve scaling similar to existing results, while requiring significantly simpler hardware capabilities.
For instance, suppose that in addition to the single miscalibration-prone device we now have access to a second device that can reliably implement $H_0$.
Supposing that we can entangle the two devices (i.e., prepare and measure in the Bell basis), this yields a trivial procedure with performance matching ours, save for the dependence on $n$.
By comparison, our protocols do not rely on entangled state preparation and measurements and do not need a second trusted device or even the ability to perform a SWAP-test between two devices.

Analogous outcomes of comparable scaling despite lesser resource requirements appear when comparing our proposed algorithms to known results in the literature.
A known Hamiltonian certification protocol \cite[Theorem~4.4]{gao2025quantumhamiltoniancertification} achieves the optimal scaling for this problem, but it uses additional resources such as $n+2$ ancillary qubits, controlled evolution (including negative time), and multi-qubit gates.
An ancilla-free version of this result \cite[Theorem~5.5]{gao2025quantumhamiltoniancertification} does not use controlled evolution or multi-qubit gates, and incurs an additional term of $\mathbf{m}^{3/2}$ (where Hamiltonians are $\mathbf{m}$-sparse) which is worse than our linear $M$ dependence (with $M$ bounding the operator norm).
The algorithms in \cite{gao2025quantumhamiltoniancertification} additionally rely on using classical knowledge of $H_0$ to implement $\e^{i H_0 t}$ and Trotterization; the former of which is a reasonable assumption in the context of Hamiltonian certification, but trivially eliminates the need to consider the recalibration problem in the first place, and the latter greatly simplifying unitary discrimination by allowing addition of a driving term~\cite{childs2000quantum}.
A Heisenberg-limited certification scheme for local Hamiltonians has been proposed in \cite{lee2025optimal} without using negative time evolution, but multi-qubit operations to implement $\e^{i H_0 t}$ are still needed.
In another recent work \cite{Nakahira2023}, the author also considered the quantum change point problem for Hamiltonians and unitary channels, where an optimal success probability is provided using joint measurements over multiple copies of the channel without considering experimentally relevant constraints.
The goal of our work is fundamentally different as we are considering a more experimentally realistic scenario where only a noisy device is provided and only simple and single-copy operations are allowed.

We note that a dependence on the system size $n$ in the sample complexity and lower threshold is not a necessity of the setting, as evidenced by prior work \cite{gao2025quantumhamiltoniancertification}.
Indeed, in our case this is a byproduct of utilizing the measurement procedure proposed by \cite{gupta2025single}, and a remaining open question of said work is to either find an efficient procedure that avoids this dependence or disprove its existence.
An alternate version of our results may be attained by instead using direct fidelity estimation \cite{Flammia_2011}, which avoids the dependence on $n$ at the cost of a worse scaling in $\epsilon$.

As a closing remark, we note that while the proposed algorithms only rely on single-qubit operations, those require precise calibration due to the adaptive measurement procedure in \cite{gupta2025single}.
In practice, a small amount of error in single-qubit gates can be tolerated due to the procedure being robust, though this introduces a constant noise floor under which miscalibrations are undetectable (i.e., $\epsilon$ needs to exceed some small value dependent on the single-qubit gate imprecision).
Additionally, our procedures are not resistant to SPAM errors, an extension which we leave for future work.

\section{Hamiltonian certification protocol}\label{sec:certification}
In this section, we define the Hamiltonian certification protocol and prove associated results.
We describe the protocol in \cref{sec:Adaptive_protocol}, and provide its detailed analysis in \cref{sec:adaptive_protocol_details}.

\subsection{Adaptive protocol}\label{sec:Adaptive_protocol}
We are given the classical description of an $n$-qubit Hamiltonian $H_0$ and an unknown $n$-qubit Hamiltonian $H$ via access to its time-evolution.
For the $i$-th experiment, where $i = 1,2,\ldots,N$, we first randomly sample the input state $\ket{\psi_0}$ from the set of stabilizer product states
\begin{equation}
    \{\ket{0},\ket{1},\ket{+},\ket{-},\ket{+i},\ket{-i}\}^{\otimes n}
\end{equation}
uniformly at random, and evolve it under the experimental Hamiltonian for time $t$, where the evolution time is set to be $t = cM^{-1}$, for small constant $c$ which we discuss in further detail at the end of \cref{sec:adaptive_protocol_details}.
After the time-evolution, we then perform adaptive measurement under an adaptive basis $\mathcal{M}$.
To construct $\mathcal{M}$, we use the classically simulated state
\begin{equation}
    \ket*{\Tilde{\phi}_t} = \sum_{j=0}^l\frac{(-itH)^j}{j!}\ket{\psi_0}= e^{-it H}\ket{\psi_0} + \bigO((t M)^{l+1}),
\end{equation}
which is the $l$-th order truncation of the state $\ket{\psi_0}$ evolved under $H_0$.
Setting $l = \Theta(\log_{1/c}(2nM^2/(c\epsilon^2)))$ guarantees that $(t M)^{l+1}< t^2\epsilon^2/2n$.
We adaptively construct the measurement basis $\mathcal{M}$ to be a \textit{DT basis} \cite{gupta2025single}, in which $\ket*{\Tilde{\phi}_t}$ is always a phase state, i.e.~having equal probability of getting both measurement outcomes on each qubit in such a basis.
Using the measurement outcomes from the adaptive measurement basis $\mathcal{M}$ on the lab state $\ket{\psi_t} = e^{-it H}\ket{\psi_0}$ over the $N$ experiments, we perform a hypothesis test as described in \cref{algo:adaptive_ham_certif} to decide whether $\norm{H-H_0}_F$ is above or below the $\epsilon$-dependent thresholds.

\begin{algorithm}[!ht]
\caption{Adaptive Hamiltonian certification}
\label{algo:adaptive_ham_certif}
\KwIn{System size $n$, upper bound on the Hamiltonian operator norm $M$, classical description of the $n$-qubit $H_0$, time-evolution access to the unknown $n$-qubit $H$, threshold $\epsilon$, success probability $\delta$.}
\KwOut{\textbf{Accept} if $\norm{H-H_0}_F \leq \bigO(\epsilon/\sqrt{n})$, \textbf{Reject} if $\norm{H-H_0}_F \geq \epsilon$}

\textbf{Set} $N \gets \bigO\!\left( \frac{n M^2 \ln(1/\delta)}{\epsilon^2} \right)$;

\For{$i=1$ \KwTo $N$}{
    \textbf{Run} $x_i \gets\,$Hamiltonian certification subroutine (\cref{algo:ham_cert_sub}).
  }

\textbf{Perform} a likelihood ratio test on $\{x_i\}_1^N$ with hypotheses $\norm{H-H_0}_F = \bigO(\epsilon/\sqrt{n})$ or larger than $\epsilon$;
\end{algorithm}

The procedure for the experiment at each step $i$ is shown in \cref{algo:ham_cert_sub}, and for completeness, we include a listing of the adaptive certification algorithm of \cite{gupta2025single} in \cref{algo:state_cert_sub}.

\begin{algorithm}[!ht]
\caption{Hamiltonian certification subroutine}
\label{algo:ham_cert_sub}
\KwIn{System size $n$, upper bound on the Hamiltonian operator norm $M$, classical description of the $n$-qubit $H_0$, time-evolution access to the unknown $n$-qubit $H$.}
\KwOut{\textbf{Accept} or \textbf{Reject}}

\textbf{Sample} $\ket{\psi_0}$ from $\{\ket{0},\ket{1},\ket{+},\ket{-},\ket{+i},\ket{-i}\}^{\otimes n}$ uniformly at random;

\textbf{Set} $t = cM^{-1}$;

\textbf{Evolve} $\ket{\psi_0}$ under $H$ for time $t$ to obtain $\ket{\psi_t} = e^{-it H}\ket{\psi_0}$;

\textbf{Classically compute} $\ket*{\Tilde{\phi}_t}$;

\textbf{Run} state certification subroutine (\cref{algo:state_cert_sub}) with $\ket{\psi_t}$ and $\ket*{\Tilde{\phi}_t}$ as input;
\end{algorithm}

\begin{algorithm}[!ht]
\caption{State certification subroutine (\cite[Algorithm~1]{gupta2025single})}
\label{algo:state_cert_sub}
\KwIn{One copy of unknown quantum state $\ket{\psi}$, known classical description of quantum state $\ket{\phi}$.}
\KwOut{\textbf{Accept} or \textbf{Reject}}

\textbf{Sample} a $k\in[n]$ uniformly at random;

\textbf{Measure} the first $k-1$ qubits of $\ket{\psi}$ in the computational basis, obtaining~$x\in\{0,1\}^{k-1}$;

\textbf{Measure} the last $n-k$ qubits of $\ket{\psi^x}$ in the basis $\mathcal{M}_x$, obtaining~$\ell$, where $\mathcal{M}_x$ is the adaptively constructed basis based on $\ket{\phi^x}$.
Let $\ket{\psi^\prime}$ be the resulting $1$-qubit state; let $\ket{\phi^\prime}$ denote $\ket{\phi}$ conditioned on outcomes $x, \ell$;

\textbf{Measure} $\ket{\psi^\prime}$ in a basis containing $\ket{\phi^\prime}$;

\textbf{Output} \textsc{Accept} iff the outcome is $\ket{\phi^\prime}$;

\end{algorithm}

We show that, with $\bigO(n\ln(1/\delta)/\xi)$ sample complexity, where $\xi$ is the fidelity threshold and $\delta$ is an upper bound of the failure probability, one could certify whether or not $\ket{\psi_t}$ is $\xi$ far away from $\ket*{\Tilde{\phi}_t}$.

\subsection{Detailed analysis of the adaptive protocol}\label{sec:adaptive_protocol_details}
For each experiment, let $\rho = e^{-iH_0t}\kb{\psi_0}e^{iH_0t}$ be the hypothesis state after evolving under the target Hamiltonian $H_0$, and $\sigma = e^{-iHt}\kb{\psi_0}e^{iHt}$ be the state obtained under the actual system Hamiltonian $H$.
We aim to determine whether $\norm{H-H_0}_F$ is large or small in relation to a threshold $\epsilon$, where $\norm{\cdot}_F$ is the normalized Frobenius norm.
We reduce the problem to comparing the fidelity between $\rho$ and $\sigma$ to $1-\xi$.

The fidelity between the two pure states $\rho$ and $\sigma$ can be written as
\begin{equation}
    F(\rho,\sigma)
    = \abs{\bra{\psi_0}e^{iH_0t}e^{-iHt}\ket{\psi_0}}^2
    = \abs{\bra{\psi_0}W(t)\ket{\psi_0}}^2,
\end{equation}
where
\begin{equation}
    W(t) = e^{iH_0t}e^{-iHt}\,.
\end{equation}
Note that if $H=H_0$, then $W(t) = I$, and the fidelity between $\rho$ and $\sigma$ is always $1$.
We have
\begin{subequations}
    \begin{align}
        \abs{\bra{\psi_0}W(t)\ket{\psi_0}}^2
        &= \Tr(\bra{\psi_0}W(t)\ket{\psi_0}(\bra{\psi_0}W(t)\ket{\psi_0})^*)\\
        & = \Tr((W(t)\otimes W^\dagger(t))(\kb{\psi_0}\otimes\kb{\psi_0})).
    \end{align}
\end{subequations}
Since $\ket{\psi_0}$ is a randomly sampled stabilizer product state, it can be written as
\begin{equation}
    \kb{\psi_0} = \kb{P,b} = \bigotimes_{i=1}^n \left(\frac{I + (-1)^{b_i}P_i}{2}\right), \qquad
    \text{where}\ \ P = \bigotimes_{i=1}^n P_i\in\{X,Y,Z\}^{\otimes n}
\end{equation}
is the $n$-qubit Pauli operator, and $b = (b_1,b_2,\ldots,b_n)\in\{0,1\}^n$ as it is a pure state.
By linearity, we have that the expectation value of the fidelity over the randomly sampled input state is
\begin{equation}
    \E[F(\rho,\sigma)] = \mathbb{E}[\Tr((W(t)\otimes W^\dagger(t))(\kb{\psi_0}\otimes\kb{\psi_0}))] = \Tr((W(t)\otimes W^\dagger(t)) K),
\end{equation}
where $K = \frac{1}{6^n} \left( \sum_{P,b} \kb{P,b} \otimes \kb{P,b}\right)$.
Expanding the operator $W(t)$ in the Pauli basis gives
\begin{equation}
    W(t) = \sum_{P\in\{I,X,Y,Z\}^{\otimes n}} w_P(t) P,\quad w_P(t) = \frac{1}{2^n}\Tr(PW(t)).
\end{equation}
And we can write the expected fidelity as
\begin{equation}
        \mathbb{E}(F(\rho,\sigma)) = \Tr((W(t)\otimes W^\dagger(t)) K) \sum_{Q,R} \overline{w_Q(t)} w_R(t) \Tr((Q \otimes R) K).
\end{equation}
Notice that the latter term may be expanded as
\begin{subequations}
\begin{align}
    \Tr((Q \otimes R) K)&= \frac{1}{6^n} \sum_{P,b} \bra{P,b} Q \ket{P,b} \bra{P,b} R \ket{P,b} \\
    &= \frac{1}{6^n} \sum_{P,b} \left( \prod_{i=1}^n \bra{P_i,b_i} Q_i \ket{P_i,b_i} \bra{P_i,b_i} R_i \ket{P_i,b_i} \right) \\
    &= \prod_{i=1}^n \left( \frac{1}{6} \sum_{P_i,b_i} \bra{P_i,b_i} Q_i \ket{P_i,b_i} \bra{P_i,b_i} R_i \ket{P_i,b_i} \right)\label{eq:product_3design}
\end{align}
\end{subequations}
To further simplify this, consider the quantity inside the product.
Fix arbitrary $Q,R\in\{I,X,Y,Z\}^{\otimes n}$, and let $P\in\{X,Y,Z\}^{\otimes n}$.
Because $\{\ket{P,b}\}$ forms a 3-design on each qubit individually, we have

\begin{subequations}
\begin{align}
    \frac{1}{6}\sum_{P_i,b_i} \bra{P_i,b_i}Q_i\ket{P_i,b_i}\bra{P_i,b_i}R_i\ket{P_i,b_i}
    &= \int_{\mathcal{S}(\mathcal{X})} \bra{\psi}Q_i\ket{\psi}\bra{\psi}R_i\ket{\psi} \mathrm{d}\psi \\
    &= \frac{1}{6}(\Tr(Q_i)\Tr(R_i)+\Tr(Q_iR_i)),
\end{align}
\end{subequations}
where the latter equation is due to Weingarten calculus.
One can check that the right-hand side is non-zero only if $Q_i=R_i$.
If $Q_i=R_i=I$ then the result is 1, and if $Q_i=R_i\neq I$ then the result is $1/3$.
Therefore

\begin{equation}
    \frac{1}{6} \sum_{P_i,b_i}\bra{P_i,b_i}Q_i\ket{P_i,b_i}\bra{P_i,b_i}R_i\ket{P_i,b_i} =
    \begin{cases}
        0 &\text{ if } Q_i\neq R_i,\\
        1 &\text{ if } Q_i=R_i=I, \\
        1/3 &\text{ if } Q_i=R_i\neq I.
    \end{cases}
\end{equation}

Using this identity, we see that each term in the product in \cref{eq:product_3design} is nonzero precisely when $Q_i = R_i$ for all $i$, that is, $Q = R$.
In this case the term equals $(1/3)^{\text{wt}(Q)}$ with $\text{wt}(Q)$ being the number of non-identity terms in $Q$.
The expected fidelity then simplifies to
\begin{equation}
    \mathbb{E}(F(\rho,\sigma)) = \sum_P 3^{-\mathrm{wt}(P)} \abs{w_P(t)}^2.
\end{equation}
Since we only want to test whether the fidelity is $1$ or far from $1$, we focus on analyzing the term associated with $I$ in the sum.
When $H=H_0$, $\abs{w_I(t)}^2 = 1$ and all other terms in this sum are zero.
When $H$ is far from $H_0$, $\abs{w_I(t)}^2$ is far from $1$.

We now analyze the function
\begin{equation}
    \mathbf{F}_I(t)
    = \abs{w_I(t)}^2
    = \frac{1}{4^n} \abs{\Tr[W(t)]}^2
    = \frac{1}{4^n} \abs{\Tr[e^{iH_0 t}e^{-iHt}]}^2.
\end{equation}
We proceed with a Taylor expansion of $\mathbf{F}_I(t)$.
Notice that all odd orders of $\mathbf{F}_I(t)$ are zero, and so we write $\mathbf{F}_I(t)$ as the summation of all even-order terms
\begin{equation}
    \mathbf{F}_I(t) = 1 - \frac{1}{2}t^2\mathbf{F}_I^{(2)}(t) + \frac{1}{4!}t^4\mathbf{F}_I^{(4)}(t) + \cdots.
\end{equation}
For the second-order terms, we could directly expand $e^{-iH_0t}$ and $e^{iHt}$ in $W(t)$.
Expand the target Hamiltonian and the system Hamiltonian in the Pauli basis:
\begin{equation}
    H_0 =  \sum_a \mu_{a0} P_a,\quad H = \sum_a \mu_{a1}P_a,
\end{equation}
we have
\begin{subequations}\label{eq:W_tayloy}
    \begin{align}
        W(t) &= e^{iH_0 t}e^{-iH t}\\
        &= \left(I + iH_0 t - \frac{1}{2}H_0^2t^2 + \ho{}\right)\left(I-iH t - \frac{1}{2}H^2t^2 + \ho{}\right)\\
        &= I - it(H-H_0) -\frac{t^2}{2}\left(H_0^2-2H_0H+H^2\right) + \ho{}\\
        &= I - it\sum_a(\mu_{a1}-\mu_{a0})P_a-\frac{t^2}{2}\sum_{a,b}(\mu_{a0}\mu_{b0} - 2\mu_{a0}\mu_{b1} + \mu_{a1}\mu_{b1})P_aP_b + \ho{}\\
        & = I - it\sum_a(\mu_{a1}-\mu_{a0})P_a-\frac{t^2}{2}\sum_{a,b}\left((\mu_{a1} - \mu_{a0})(\mu_{b1}-\mu_{b0}) +(\mu_{a1}\mu_{b0}-\mu_{a0}\mu_{b1})\right)P_aP_b + \ho{}\nonumber\\
        & = I - it\sum_a(\mu_{a1}-\mu_{a0})P_a-\frac{t^2}{2}\sum_a (\mu_{a1}-\mu_{a0})^2I - \frac{t^2}{2}\sum_{a\neq b}(\mu_{a1}\mu_{b0}-\mu_{a0}\mu_{b1})P_aP_b + \ho{}\nonumber\\
        & = \left(1-\frac{t^2}{2}\sum_a (\mu_{a1}-\mu_{a0})^2\right)I - it\sum_a(\mu_{a1}-\mu_{a0})P_a - \frac{t^2}{2} \hspace*{-1em}\sum_{\substack{a\neq b\\ [P_a,P_b]\neq 0}}\hspace*{-0.8em}(\mu_{a1}\mu_{b0}-\mu_{a0}\mu_{b1})P_aP_b + \ho{},\nonumber
    \end{align}
\end{subequations}
when $P=I$, the entanglement fidelity of $W(t)$ with identity is
\begin{equation}
    \abs{w_I(t)}^2 = \left(1-\frac{t^2}{2}\sum_a (\mu_{a1}-\mu_{a0})^2\right)  + \ho{} = 1 - t^2 \norm{H_0-H}_F^2 + \ho{}
\end{equation}

In order to obtain a bound on the high order terms $\ho{}$, we first analyze $W(t)$.
We first introduce some notation that greatly simplifies our derivation.

\begin{definition}[Normal ordering]
    For a monomial $A$ of $H_0$ and $H$, we define its normal ordering $\no{A}$ as the monomial obtained by moving all $H_0$ to the left.
    For example, $\no{H_0 H^2 H_0}=H_0^2 H^2$.
    We extend this definition to all polynomials of $H_0$ and $H$ by linearity and to all their limits by continuity.
\end{definition}

We can then write down the $k$th-order derivatives of $W(t)$ as
\begin{equation}\label{eq:derivatives_of_W(t)}
    W^{(k)}(t) = i^k e^{iH_0 t}\no{(H_0-H)^k}e^{-iHt} = i^k \no{e^{iH_0 t}(H_0-H)^k e^{-iHt}}.
\end{equation}

For $\no{e^{iH_0 t}(H_0-H)^k e^{-iHt}}$, we can use the cyclic property of the trace to show that
\begin{subequations}\label{eq:trace_normal_ordering_cyclic}
\begin{align}
    &\Tr[\no{e^{iH_0 t}(H_0-H)^k e^{-iHt}}] \\
    &= \Tr[H_0e^{iH_0 t}\no{(H_0-H)^{k-1}} e^{-iHt}]-\Tr[e^{iH_0 t}\no{(H_0-H)^{k-1}} e^{-iHt}H] \\
    &= \Tr[e^{-iHt}H_0e^{iH_0 t}\no{(H_0-H)^{k-1}} ]-\Tr[ e^{-iHt}H e^{iH_0 t}\no{(H_0-H)^{k-1}}] \\
    &= \Tr[e^{-iHt}(H_0-H) e^{iH_0 t} \no{(H_0-H)^{k-1}}].
\end{align}
\end{subequations}
Moreover, we have
\begin{equation}\label{eq:second_order_frobenius_bound}
    \norm{\no{(H_0-H)^2}}_F \leq \norm{H_0(H_0-H)}_F + \norm{(H_0-H)H}_F\leq (\norm{H_0}+\norm{H})\norm{H_0-H}_F,
\end{equation}
and
\begin{subequations}\label{eq:third_order_frobenius_bound}
    \begin{align}
        \norm{\no{(H_0-H)^3}}_F &\leq \norm{H_0\no{(H_0-H)^2}}_F + \norm{\no{(H_0-H)^2}H}_F \\
        &\leq (\norm{H_0}+\norm{H})\norm{\no{(H_0-H)^2}}_F\leq (\norm{H_0}+\norm{H})^2\norm{H_0-H}_F.
    \end{align}
\end{subequations}
The above inequalities allow us to bound the derivatives of $f(t)$.
First, we have $\abs{w_I(t)} \leq 1$.
For $w_I'(t)$,
\begin{equation}
    \abs{w_I'(t)}
    = \frac{1}{2^n} \abs{\Tr[e^{-iHt}e^{iH_0 t}(H_0-H)]}
    \leq \norm{H_0-H}_F \norm{e^{-iHt}e^{iH_0 t}}_F
    = \norm{H_0-H}_F
\end{equation}
using the Cauchy-Schwarz inequality for the normalized Hilbert-Schmidt inner product.
For $w_I''(t)$, we use \eqref{eq:trace_normal_ordering_cyclic} to get
\begin{subequations}
\begin{align}
    \abs{w_I''(t)}
    &= \frac{1}{2^n} \abs{\Tr[e^{-iHt}(H_0-H) e^{iH_0 t}(H_0-H)]} \\
    &\leq \norm{(H_0-H)e^{iHt}}_F \norm{e^{iH_0 t}(H_0-H)}_F \\
    &= \norm{H_0-H}_F^2,
\end{align}
\end{subequations}
where we first used the Cauchy-Schwarz inequality and then the fact that the normalized Frobenius norm is invariant under multiplication with a unitary.
For $w_I^{(3)}(t)$, we similarly have
\begin{subequations}
\begin{align}
    \abs*{w_I^{(3)}(t)}
    &= \frac{1}{2^n} \abs{\Tr[e^{-iHt}(H_0-H) e^{iH_0 t}\no{(H_0-H)^2}]} \\
    &\leq \norm{H_0-H}_F \norm*{\no{(H_0-H)^2}}_F \\
    &\leq (\norm{H_0}+\norm{H}) \norm{H_0-H}_F^2.
\end{align}
\end{subequations}
This also applies to the fourth-order derivative, for which we have
\begin{subequations}
\begin{align}
    \abs*{w_I^{(4)}(t)}
    &= \frac{1}{2^n} \abs{\Tr[e^{-iHt}(H_0-H) e^{iH_0 t}\no{(H_0-H)^3}]} \\
    &\leq \norm{H_0-H}_F \norm*{\no{(H_0-H)^3}}_F \\
    &\leq (\norm{H_0}+\norm{H})^2\norm{H_0-H}_F^2.
\end{align}
\end{subequations}

We are now ready to bound $\mathbf{F}_I^{(4)}(t)$ through
\begin{equation}
    \mathbf{F}_I^{(4)}(t) = w_I^*(t)w_I^{(4)}(t) + 4(w_I'(t))^* w_I^{(3)}(t) + 6 (w_I''(t))^* w_I''(t) + 4(w_I^{(3)}(t))^*w_I'(t) + (w_I^{(4)}(t))^*w_I(t).
\end{equation}
We therefore have
\begin{equation}
    \abs*{\mathbf{F}_I^{(4)}(t)}
    \leq 2(\norm{H_0}+\norm{H})^2 \norm{H_0-H}_F^2 + 8(\norm{H_0}+\norm{H}) \norm{H_0-H}_F^3 + 6\norm{H_0-H}_F^4.
\end{equation}
By this bound together with $\norm{H},\norm{H_0} \leq M$, we have
\begin{equation}
    \abs*{\mathbf{F}_I^{(4)}(t)} = \bigO(M^2\norm{H_0-H}_F^2),
\end{equation}
since the threshold $\epsilon$ for $\norm{H_0-H}_F^2$ is small compared to $M$.
Therefore,
\begin{subequations}\label{eq:F_I}
\begin{align}
    \mathbf{F}_I(t)
    &= 1 - t^2 \norm{H_0-H}_F^2 + \bigO\biggl(\sum_{j = 2}^{\infty}t^{j+2}M^j \norm{H_0-H}_F^2\biggr) \\
    &= 1 - t^2 \norm{H_0-H}_F^2 + \bigO\left(t^4M^2\norm{H_0-H}_F^2\right).
\end{align}
\end{subequations}
On the one hand, if $H=H_0$, then $\mathbf{F}_I(t) =1$.
On the other, $\norm{H_0-H}_F=\epsilon$ means that $H$ is far from $H_0$ over a threshold $\epsilon$, and so
\begin{equation}
    \mathbf{F}_I(t) = 1 - t^2 \epsilon^2 + \bigO\left(t^4M^2\epsilon^2\right).
\end{equation}

With this, we can write the fidelity between $\rho$ and $\sigma$ as
\begin{subequations}
\begin{align}\label{eq:fidelity-bound}
    \mathbb{E}(F(\rho,\sigma))
    &= \sum_P \abs{w_P(t)}^2 3^{-wt(P)}
    = \abs{w_I}^2 + \sum_{P\neq I} \abs{w_P(t)}^2 3^{-wt(P)}\\
    &\leq \abs{w_I(t)}^2 + \frac{1}{3}\sum_{P\neq I} \abs{w_P(t)}^2
    = \abs{w_I(t)}^2 + \frac{1}{3}(1-\abs{w_I(t)}^2)\\
    &=\frac{1}{3}+\frac{2}{3} \abs{w_I(t)}^2
    = 1-\frac{2}{3}t^2\epsilon^2 + \bigO(t^4M^2\epsilon^2),
\end{align}
\end{subequations}
where the inequality comes from the fact that $wt(P)\geq 1,\forall P\neq I$.
If we set $t = cM^{-1}$ with an appropriate constant $c$, when $\norm{H_0 - H}_F \geq \epsilon$ we could upper bound the fidelity by $1-\xi$, where $\xi = \Omega(\epsilon^2/M^2)$.

With the same evolution-time parameters chosen, we can also lower bound the expected infidelity as
\begin{subequations}
\begin{align}\label{eq:in-fidelity-bound}
    \mathbb{E}(1-F(\rho,\sigma))
    &= 1 - \sum_P \abs{w_P(t)}^2 3^{-wt(P)}
    = 1 - \abs{w_I}^2 - \sum_{P\neq I}\abs{w_P(t)}^2 3^{-wt(P)}\\
    &\leq 1 - \abs{w_I(t)}^2
    = t^2 \norm{H_0-H}_F^2 + \bigO\left(t^4M^2\norm{H_0-H}_F^2\right),
\end{align}
\end{subequations}
where the last equality is by \eqref{eq:F_I}.
With the same $t = cM^{-1}$, if $\norm{H_0-H}_F<\bigO\left(\epsilon/\sqrt{n}\right)$, then the fidelity can be lower-bounded by $1-\xi/n$, where $\xi = \bigO(\epsilon^2/M^2)$.
By \cite{gupta2025single}, we can distinguish between $F(\rho,\sigma) > 1 - \xi/n$ or $F(\rho,\sigma)<1-\xi$ with probability at least $1-\delta$ by $\bigO(n\ln(1/\delta)/\xi)$ samples of the state.
So, to solve the single-device Hamiltonian certification problem as in \cref{def:Ham_certification_prob} with probability at least $1-\delta$, the sample complexity and the total evolution time complexity are
\begin{equation}
    N = \bigO\left(\frac{nM^2\ln(1/\delta)}{\epsilon^2}\right),\quad T = Nt = \bigO\left(\frac{nM\ln(1/\delta)}{\epsilon^2}\right).
\end{equation}

\section{Hamiltonian \changepoint{} detection}

In deriving the Hamiltonian certification procedure in \cref{sec:certification} above, we have presented its results in terms of sample complexity when the unknown Hamiltonian $H$ remains constant.
In practice (and more so in our targeted recalibration setting), several additional details need to be considered:
\begin{itemize}
    \item The unknown Hamiltonian $H$ can be time-dependent.
    When performing several rounds of evolution by the same time $t$, we may effectively obtain evolution by differing Hamiltonians $H_1, H_2, \ldots$.
    \item The change in $H$ can happen anywhere in the above sequence; as such, algorithms should give more priority to more recent data.
    \item Knowledge about a detected change needs to be known as soon as possible, meaning that algorithms need frequent opportunities to terminate.
\end{itemize}
These details are characterized by the \changepoint{} setting, introduced previously in \cref{def:ham_changepoint}.
The goal of this section is to derive the results in this setting as stated in \cref{thm:online}, which may be viewed as an alternate means to present earlier certification results.
In particular, the analysis here relies on the single-sample Hamiltonian certification procedure in \cref{algo:ham_cert_sub} with a different method for aggregating its classical statistics.
The method and analysis for the latter follow from standard results, briefly introduced below.

\subsection{Cumulative sum (CUSUM)}

Our earlier certification algorithm performs a hypothesis test between two possibilities.
Specifically, \cref{algo:adaptive_ham_certif} may be viewed as a likelihood ratio test which rejects when the alternate hypothesis appears to be at least $1/\delta$ times more likely than the null.
Generalizing this to the \changepoint{} setting involves keeping track of alternate hypotheses for each change location, terminating at step $i$ if
\begin{equation}
    \frac{q(x_{\nu+1})}{p(x_{\nu+1})} \cdot \frac{q(x_{\nu+2})}{p(x_{\nu+2})} \cdots \frac{q(x_i)}{p(x_i)} \geq \exp(h) \qquad \text{for any}\ \nu \in \{0, \ldots, i\},
\end{equation}
where $p$ and $q$ are the densities of the pre-change and post-change distributions, and $x_1, x_2, \ldots$ is the sequence of observations (for clarity, throughout this section we denote random variables with uppercase letters and their observations with corresponding lowercase ones).
This ensures that one of the change locations is at least $\exp(h)$-times more likely than no change, and by frequently checking this condition the algorithm terminates as soon as it is satisfied.
Knowing for which value of $\nu$ the above condition rejects additionally tells us the maximum likelihood estimate for the change location.
By taking the logarithm, the above conditions are equivalent to
\begin{equation}
    \sum_{j=\nu+1}^i \ln \frac{q(x_j)}{p(x_j)} \geq h \qquad \text{for any}\ \nu \in \{0, \ldots, i\},
\end{equation}
and we can further simplify this to checking a single inequality by keeping track of a non-negative cumulative sum as seen in \cref{algo:cusum} below.

\begin{algorithm}[htbp]
\caption{Cumulative sum (CUSUM) \cite{page1954continuous}}
\label{algo:cusum}
\KwIn{Known distributions $P, Q$, log-likelihood threshold $h$, observations $x_1,x_2,\ldots$ of independent $X_1,X_2,\ldots$ (received sequentially).}
\KwOut{\textbf{Continue} or \textbf{Terminate}}

\textbf{Set} $s_0 \gets 0$;

\For{$i=1,2,\ldots$}{
    \textbf{Receive} $x_i$;

    \textbf{Compute} $z_i \gets \ln \frac{dQ}{dP} (x_i)$;
    \hfill\CommentSty{// Log-likelihood ratio}

    \textbf{Set} $s_i \gets \pos{s_{i-1} + z_i}$;

    \eIf{$s_i \geq h$}{
        \textbf{Terminate};
        \hfill\CommentSty{// \Changepoint{} detected; stop}
    } {
        \textbf{Continue};
        \hfill\CommentSty{// Continue the loop on line 2}
    }
}
\end{algorithm}
Note that $\frac{dQ}{dP}(x_i)$ in \cref{algo:cusum} is the likelihood ratio as given by the Radon-Nikod{\'y}m derivative \cite{shiryaev1996minimax}, which equals $q(x_i) / p(x_i)$ when distributions $P,Q$ admit density or mass functions $p,q$.
The previously mentioned estimator for the change location can be retrieved by finding the largest $\nu$ such that $s_\nu = 0$.

\paragraph{Performance and optimality:}
The CUSUM procedure detects changes with small delay; its ability in this respect has been proven to be asymptotically optimal \cite{lorden1971procedures}.
To interpret the corresponding result, let us first properly define the used metric of detection delay.
Recall that $\E_\nu$ conditions on a change happening at location $\nu$, and $\E_\infty$ on no change happening.
\begin{definition}\label{def:delay}
    The Lorden detection delay \cite{lorden1971procedures}, \cite[Eq.~6.17]{tartakovsky2014sequential} is
    \begin{equation}\label{eq:delay}
        \delay{\pos{N-\nu}} = \sup_{\nu \geq 0} \esssup \E_\nu [(N-\nu)^+ \mid H_1, \ldots, H_\nu ],
    \end{equation}
    where maximization is first done over trajectories up to the \changepoint{}, and then its location.
\end{definition}
This notion of delay can be seen as a worst-case metric, and minimizing over it leads to a minimax problem.
In considering \changepoint{} detection over classical distributions, we use prior definitions of metrics with random variables $X_i$ in place of Hamiltonians $H_i$.

One may wonder the purpose of defining detection delay $\delay{\pos{N-\nu}}$ as in \eqref{eq:delay} instead of using a simpler quantity such as $\E_0[N]$ (i.e., assuming that the \changepoint{} happens immediately), considering that the two are equivalent for CUSUM.
The issue here is that small $\E_0[N]$ does not necessarily imply fast detection for later \changepoint{} locations, as $\E_\nu[N]$ can in general increase with $\nu$.
Indeed, it is trivial to construct a strategy (by removing the non-negativity restriction) having false termination time $\falsealarm{N} = \infty$ and very small $\E_0[N]$, yet one where detection delay is $\delay{\pos{N-\nu}} = \infty$ \cite{lorden1971procedures}.

\begin{lemma}[CUSUM optimality \cite{lorden1971procedures}]\label{lem:cusum}
    The detection delay of any stopping rule satisfies
    \begin{equation}\label{eq:cusum}
        \delay{\pos{N-\nu}} \geq \frac{\ln \falsealarm{N}}{D(Q \| P)} \cdot (1 + o(1)).
    \end{equation}
    Moreover, the CUSUM procedure is asymptotically optimal, as it attains a version of \eqref{eq:cusum} with equality.
\end{lemma}
\Cref{eq:cusum} and subsequent \changepoint{} results give the scaling of the detection delay $\delay{\pos{N-\nu}}$ with average run length $\falsealarm{N}$ when no change occurs.
Optimality follows from \cite[Theorem~3]{lorden1971procedures}, and achievability from \cite[Theorem~2]{lorden1971procedures} and the subsequent remark.
In particular, the latter theorem equates CUSUM to a series of parallel open-ended Sequential Probability Ratio Tests (SPRTs) \cite{wald1945sequential} and leverages known results \cite{wald1947sequential} for their analysis.

\Cref{lem:cusum} is primarily a statement about the scaling of the delay $\delay{\pos{N-\nu}}$ with respect to the false-alarm length $\falsealarm{N}$, where the latter can be exponentially larger than the former.
In CUSUM, this tradeoff is facilitated by the threshold parameter $h$, as we have $\falsealarm{N} \geq \exp(h)$ and \eqref{eq:cusum} holds with $h$ in place of $\ln \falsealarm{N}$.
In plain terms, such a tradeoff is made possible by the cumulative score $z_i$ increasing on average under $Q$, but the same steady increase being exponentially unlikely under $P$.

\paragraph{Robustness:}
The requirement of needing the distributions $X_1, X_2, \ldots$ to be i.i.d.~(aside from changing at the \changepoint{}) and exactly known is rather strict.
Indeed, any uncertainty in either the Hamiltonian norm or the certification procedure leads to a range of possible distributions for each $X_i$, which are independent but not necessarily identical.
Specifically, we need to consider cases where the distributions can vary with time $i$ and the \changepoint{} location $\nu$,
\begin{itemize}[leftmargin=8em]
    \item[(\Changepoint{}):] $X_1, X_2, \ldots \sim P_1^{(\nu)}, \ldots, P_\nu^{(\nu)}, Q_{\nu+1}^{(\nu)}, Q_{\nu+2}^{(\nu)}, \ldots$\ ,
    \item[(No \changepoint{}):] $X_1, X_2, \ldots \sim P_1, P_2, \ldots$\ .
\end{itemize}
Instead of knowing these exactly, we instead have that pre-change distributions belong to some class $P_i, P_i^{(\nu)} \in \mathcal{P}$ and the post-change ones to another class $Q_i^{(\nu)} \in \mathcal{Q}$.

Fortunately, known approaches in literature \cite{unnikrishnan2011minimax,molloy2017misspecified,liang2022quickest,molloy2019minimax} can show that CUSUM is robust in this scenario.
Suppose that the distribution classes can be separated via $\mathcal{P} \preceq \overline{P} \preceq \underline{Q} \preceq \mathcal{Q}$, then $\overline{P}$ and $\underline{Q}$ are least-favorable distributions (LFDs) that can be used as parameters in CUSUM and its analysis.
Here, $A \succeq B$ indicates that the likelihood ratio $\frac{dQ}{dP}(X)$ (or equivalently, its logarithm $Z$) is stochastically larger \cite{veeravalli1994minimax} under $A$ than under $B$, meaning
\begin{equation}\label{eq:stochastic-order}
    \Pr[Z > x \mid X \sim A] \geq \Pr[Z > x \mid X \sim B] \qquad \forall x \in \mathbb{R}.
\end{equation}

In simple terms, replacing $P$ by some $P_i \preceq P$ at any one timestep by definition can only reduce the scores at that step, increasing the false-alarm time $\falsealarm{N}$.
Similarly, replacing $Q$ by some $Q_i \succeq Q$ increases scores and shortens the detection delay $\delay{\pos{N-\nu}}$.
As such, it suffices to analyze CUSUM for a fixed choice of LFDs, which provides run length bounds for the general case.
Note lastly that in the case of dealing with Bernoulli distributions $\Ber(\theta)$ that produce `1' with probability $\theta$ and `0' with probability $1-\theta$, $\Ber(a) \succeq \Ber(b)$ is equivalent to $a \geq b$ (so long as $Z(1) \geq Z(0)$).

For completeness, we provide the exact used statement and a self-contained proof in \cref{appx:cusum-robust}, noting that this result was previously shown in \cite[Theorem~1]{molloy2019minimax}.

\subsection{Hamiltonian \changepoint{} algorithm and analysis}

It remains to apply CUSUM to our single-sample subroutine for Hamiltonian certification in \cref{algo:ham_cert_sub}.
In particular, we denote the outcomes of this subroutine as $x_i$.
The resulting procedure is shown in \cref{algo:ham_changepoint}.

\begin{algorithm}[htbp]
\caption{Hamiltonian \changepoint{} detection}
\label{algo:ham_changepoint}
\KwIn{System size $n$, upper bound on the Hamiltonian operator norm $M$, classical description of the $n$-qubit $H_0$, time-evolution access to the unknown $n$-qubit $H$, threshold $\epsilon$, log-likelihood threshold $h$.}
\KwOut{\textbf{Continue} while $\norm{H-H_0}_F \leq \bigO(\epsilon/\sqrt{n})$, \textbf{Terminate} if $\norm{H-H_0}_F \geq \epsilon$}

\textbf{Set} $\xi \gets \Omega(\epsilon^2 / M^2)$;

\textbf{Initialize} CUSUM (\cref{algo:cusum}) with $P \gets \Ber(\xi/2n), Q \gets \Ber(\xi/n), h$;

\For{$i=1,2,\ldots$}{
    \textbf{Run} $x_i \gets\,$ Hamiltonian certification subroutine (\cref{algo:ham_cert_sub});

    \textbf{Send} $x_i$ to CUSUM;

    \eIf{\textup{CUSUM terminates}}{
        \textbf{Terminate};
        \hfill\CommentSty{// \Changepoint{} detected; stop}
    } {
        \textbf{Continue};
        \hfill\CommentSty{// Continue the loop on line 3}
    }
}
\end{algorithm}
Performance analysis for this algorithm combines our earlier derivations with aforementioned properties about CUSUM bounds and robustness.

\begin{proof}[Proof of \cref{thm:online}]
    Recall that by \cref{def:ham_changepoint} of the Hamiltonian \changepoint{} problem, we need to distinguish between the following cases: the no-change case with $\norm{H_0 - H}_F \leq C \epsilon/\sqrt{n}$ for all $i$ (where $C$ is a universal constant that we choose later), and the \changepoint{} case with $\norm{H_0 - H}_F \geq \epsilon$ following the \changepoint{}.
    From the discussion following \eqref{eq:fidelity-bound} and \eqref{eq:in-fidelity-bound}, the subroutine in \cref{algo:ham_cert_sub} has state infidelity at least $\xi$ in the \changepoint{} case with $\xi = \Omega(\epsilon^2/M^2)$; by choosing a small enough $C$ we can ensure infidelity is at most $\xi/2n$ in the no-change case.
    Correspondingly, the adaptive measurement procedure in \cref{algo:state_cert_sub} has rejection probability at least $\xi/n$ and at most $\xi/2n$ \cite[Theorem~1]{gupta2025single}.
    That is, we have
    \begin{align}
        \mathcal{P} &= \{\Ber(\theta) : \theta \leq \xi / 2n\},&
        \mathcal{Q} &= \{\Ber(\theta) : \theta \geq \xi / n\},
    \end{align}
    and by CUSUM robustness it suffices to analyze the case with $P = \Ber(\xi/2n)$ and $Q = \Ber(\xi/n)$.

    We proceed with a linearization bound for the KL divergence $D(Q \| P) \geq k \xi / n$ with $k = \ln(2) - 1/2$, which can be seen from its Taylor expansion.
    By the results of \cref{lem:cusum}, CUSUM then satisfies
    \begin{equation}
        \delay{\pos{N-\nu}}
        = \frac{\ln \falsealarm{N}}{D(Q \| P)} \cdot (1 + o(1))
        \leq \frac{n \ln \falsealarm{N}}{k \xi} \cdot (1 + o(1))
        = \bigO\left( \frac{n M^2 \ln \falsealarm{N}}{\epsilon^2} \right).
    \end{equation}
    A similar result for total evolution time then follows from the definitions $\delay{\pos{T - t\nu}} = t \delay{\pos{N - \nu}}$ and $\falsealarm{T} = t \falsealarm{N}$ in \cref{def:ham_changepoint}, as well as $t = c M^{-1}$.
\end{proof}

\section{Numerical experiments}

This section supplements the theoretical analysis of the proposed algorithms with numerical simulations.
In doing so, we aim to demonstrate their effectiveness by showcasing a targeted use case.

In the following experiments, we fix the target Hamiltonian to be the Rydberg Hamiltonian,
\begin{equation}
    H_0 = \frac{\Omega}{2}\sum_i X_i - \Delta \sum_i N_i + \Omega \sum_{i<j} \left(\frac{R_b}{a\abs{i-j}}\right)^6 N_iN_j,
\end{equation}
where $N_i \coloneqq \kb{r_i}$ is the projector onto the Rydberg state with $Z_i = \kb{g_i} - \kb{r_i}$, and $\{\Omega, \Delta, R_b, a\}$ are constants that we specify below.
Computations are carried out using the \textsc{qiskit}~\cite{qiskit} package.
All code used for the numerical simulations is publicly available on GitHub~\cite{github_repo}.

\paragraph{Hamiltonian certification:}

We plot the results of running \cref{algo:adaptive_ham_certif} and corresponding certification verdicts for different $\norm{H - H_0}_F$ in the left and right panels of \cref{fig:hamiltonian_certification}, respectively.
The Hamiltonian is not certified when the fraction of negative algorithm outcomes is greater than $10^{-4}$.

\begin{figure}[htbp]
    \centering
    \includegraphics[width=0.8\textwidth]{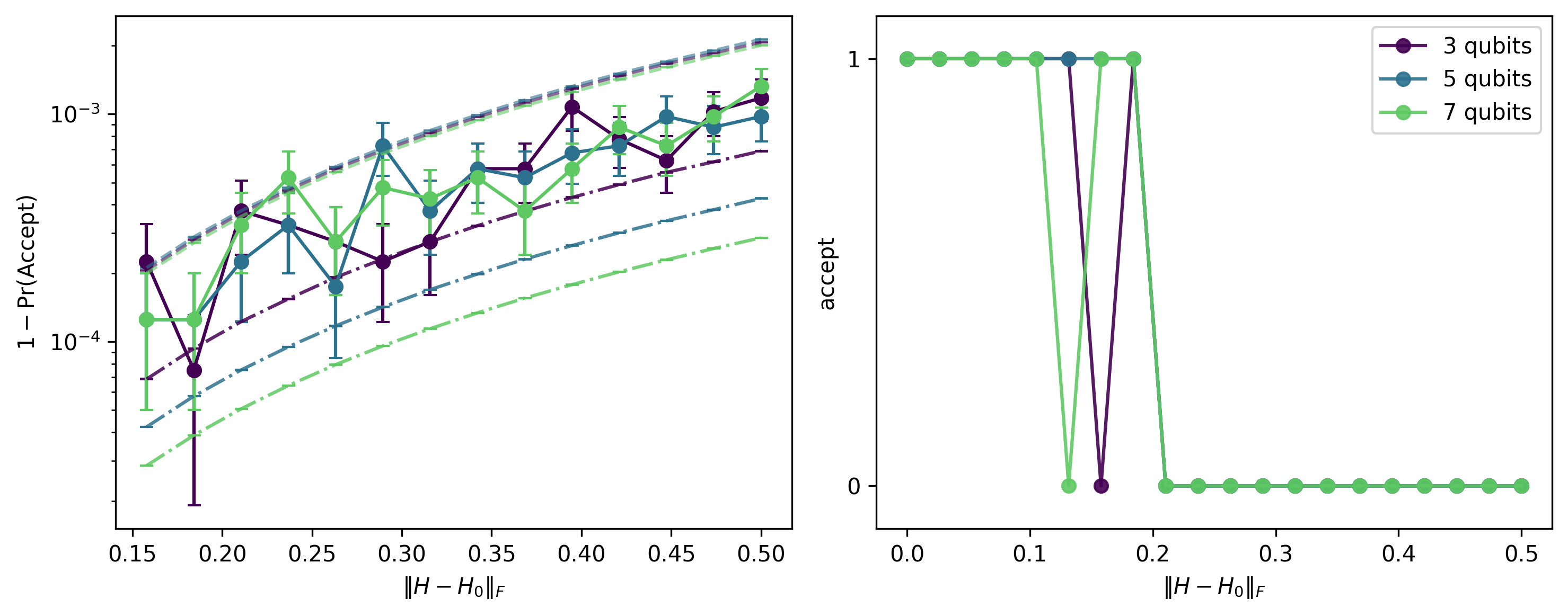}
    \caption{%
        (left) Results of running \cref{algo:adaptive_ham_certif} given by the probability of accepting the Hamiltonian $H$ on logarithmic scale.
        (right) Corresponding certification verdicts.
        The Hamiltonian is rejected when the fraction of negative outcomes is greater than $10^{-4}$.
        The acceptance probabilities are calculated as the number of positive outcomes in the $2 \times10^4$-element array we obtain for each pair $(n, \norm{H-H_0}_F)$.
        The error bars correspond to one Wilson interval.
        $\norm{\cdot}_F$ denotes the normalized Frobenius norm.
    }
    \label{fig:hamiltonian_certification}
\end{figure}

We consider systems with $n=3$, 5, and 7 qubits, fixing $\Omega=1, \Delta=2.5, R_b=1.5 $, and $a=1$.
The evolution time $t$ is set to $0.1$ throughout.
Each data point in \cref{fig:hamiltonian_certification} is a result of running \cref{algo:state_cert_sub} with $N=2 \times 10^4$.
We consider $20$ distinct values of $\norm{H-H_0}_F$ for each $n$, which amounts to the total of $4 \times 10^5$ experiments for every system size.
For every $n$, we fix a Hermitian perturbation of unit norm and add its scalar multiple to $H_0$, thus controlling $\norm{H-H_0}_F$.
We then evolve a random stabilizer product state $\ket{\psi_0}$ under $e^{-i t H_0}$ and $e^{-i t H}$, and run \cref{algo:state_cert_sub} with $\ket{\text{hyp}} = e^{-i t H_0} \ket{\psi_0}$ and $\ket{\text{lab}} = e^{-i t H} \ket{\psi_0}$.

From \cref{fig:hamiltonian_certification}, we observe that our results are consistent with the lower bounds for the acceptance and rejection probabilities of \cref{algo:state_cert_sub} as given in~\cite[Theorem~1]{gupta2025single}.
We also observe that the algorithm performs similarly across system sizes, suggesting that the rejection lower bound in~\cite{gupta2025single} could potentially be made independent of $n$ for typical instances.
From the right panel of~\cref{fig:hamiltonian_certification}, we see that the algorithm confidently detects perturbations of size $\ge 0.2$.
This is consistent with~\eqref{eq:fidelity-bound} since $\frac{2}{3}t^2 \norm{H-H_0}_F^2$ with $\norm{H-H_0}_F = 0.2$ and $t=0.1$ is on the order of the used threshold, $10^{-4}$.

\paragraph{Hamiltonian \changepoint{} detection:}

\Cref{fig:cusum_example} shows example executions of the Hamiltonian \changepoint{} procedure.
Here, it is applied to a Hamiltonian that perpetually accumulates small changes, but has a possibility of experiencing large disturbances for a small period of time.
Correspondingly, the normalized Frobenius norm either stays relatively low, or rapidly increases to a higher value.
When such a change occurs, the cumulative value of the CUSUM subroutine begins to slowly but steadily increase, as shown in the left-hand plot.
Otherwise, this value stays relatively low, as in the right-hand plot.
A key takeaway of this example is that we can easily pick a cumulative value threshold (e.g., $h = 3$) that is crossed in the left plot but not the right one.
This ensures both that changes are detected rapidly, and that the rate of false alarms is minimal.

\begin{figure}[htbp]
    \centering
    \includegraphics[width=\textwidth]{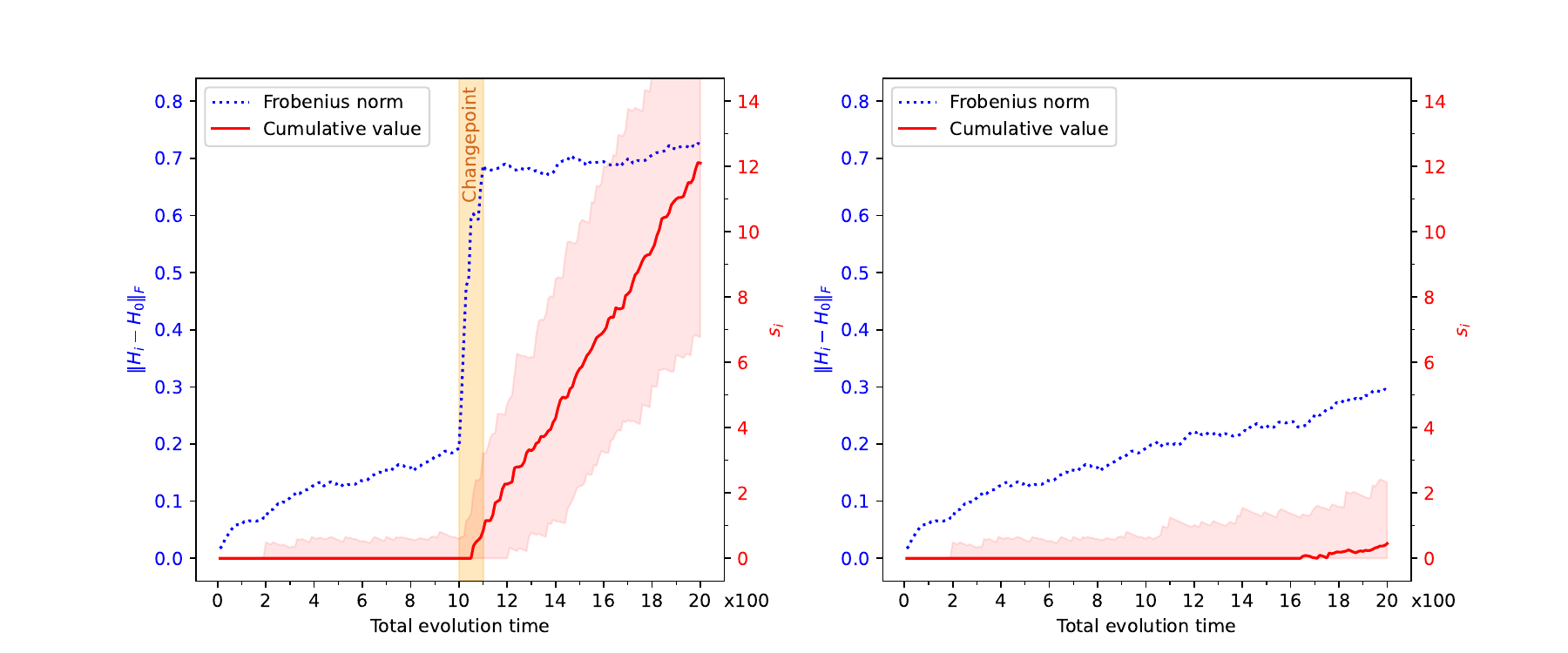}
    \caption{%
        Results of running \cref{algo:ham_changepoint} on a sequence of Hamiltonians that experiences
            (left) a significant change and
            (right) only small changes,
        with the blue dotted series showing Hamiltonian deviation over time.
        The red series is the cumulative sum $s_i$ aggregated by CUSUM (median and $95\%$ interpercentile range across $120$ trials).
    }
    \label{fig:cusum_example}
\end{figure}

The experiment in \cref{fig:cusum_example} was carried out with the parameters $n=3$ and $\xi=0.002$.
Hamiltonian perturbations are sampled from the Gaussian unitary ensemble, scaled by a factor of $0.1$ during the $1$-unit-wide \changepoint{} window and by $0.01$ outside of it.
To manage small rejection probabilities, $s=100$ shots were used for each experiment, which is accounted for in the figure via the multiplicative factor on the total evolution time axis.
In particular, each experiment uses the same stabilizer product state $\ket{\psi_0}$ but repeats the state certification subroutine in \cref{algo:state_cert_sub} independently $s$ times (meaning the values of $k$, $x$, and paths taken through the adaptive basis $\mathcal{M}_x$ can differ).
CUSUM is then performed with binomial distributions instead of Bernoulli, where the log-likelihood ratio $\ln \frac{dQ}{dP}(x_i)$ in \cref{algo:cusum} evaluates to $x_i \ln\bigl(\frac{q}{p}\bigr) + (s-x_i) \ln\bigl(\frac{1-q}{1-p}\bigr)$ when $x_i$ rejections are obtained.

In considering specifically the termination time of such CUSUM runs, we obtain the plots in \cref{fig:cusum_arl}.
The left-hand plot shows the expected sequence length sampled by \cref{algo:cusum} before termination, as a function of the actual rejection probability.
The dashed vertical lines show two hypotheses for this rejection probability; focusing on how the termination threshold $h$ affects the run length at these lines illustrates a point previously made in \cref{lem:cusum}.
As the average run length scales approximately linearly with $h$ at the $q$ line, the individual series (each incrementing $h$ by $5$) become increasingly crowded on the log-scale for larger values of $h$.
These series are spread out at the $p$ line, indicating roughly exponential scaling with $h$.%

The right-hand plot of \cref{fig:cusum_arl} demonstrates a similar principle of larger deviations resulting in shorter run lengths when applied to Hamiltonian change detection.
Error bars on the plot capture variability over several choices of $H$, as perturbations having the same norm difference $\norm{H-H_0}_F$ yield a range of possible rejection probabilities.
Note that the intervals shown for the smallest and largest norm values do not overlap, meaning these cases are distinguishable even when the perturbed Hamiltonian $H$ is chosen adversarially.
In the context of this plot, $H$ remains fixed throughout each CUSUM run.
Expectation is over all variables except $H_0$ and $H$---that is, the initial state $\ket{\psi_0}$, the choice of $k$, and all measurement outcomes (and thus the computational basis outcome $x$, the path through the adaptive basis $\mathcal{M}_x$, and the CUSUM scores sequence).

\begin{figure}[htbp]
    \centering
    \includegraphics[width=\textwidth]{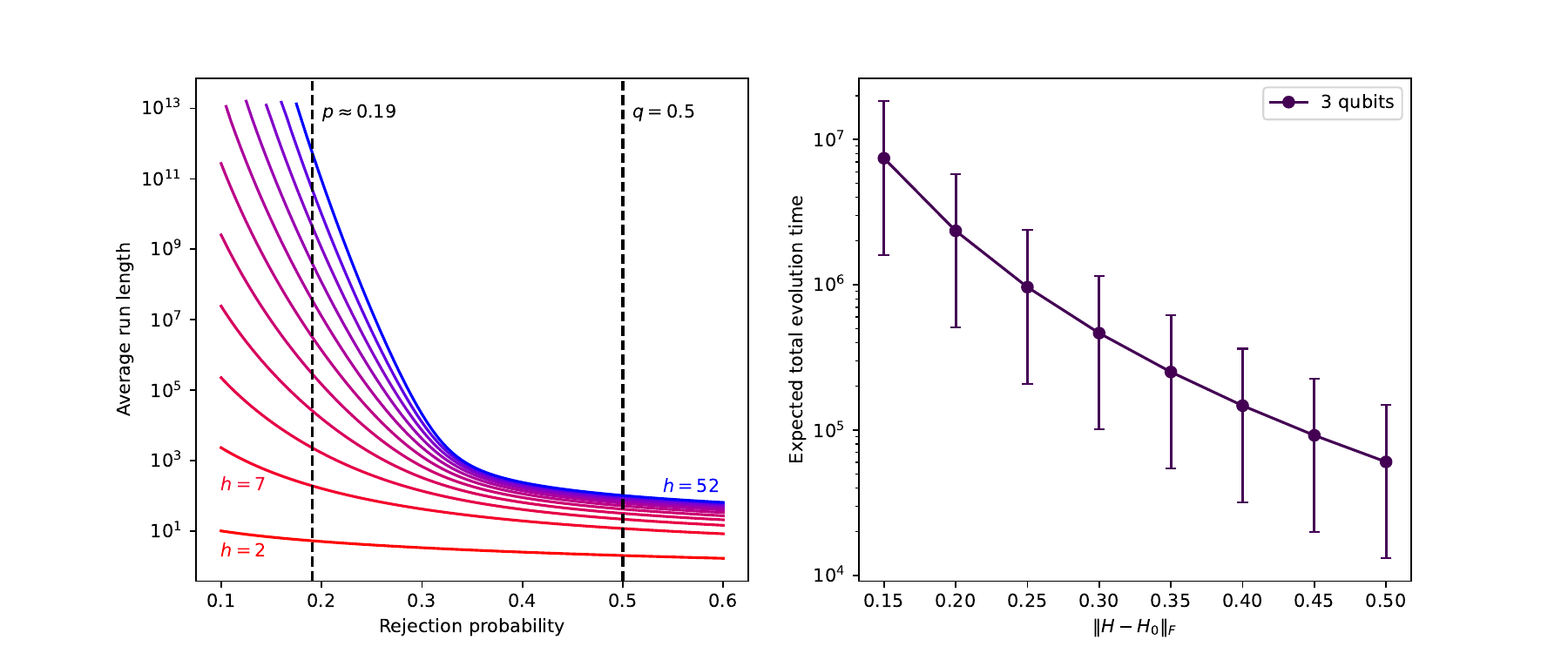}
    \caption{%
        Expected termination time of CUSUM when
            (left) run directly, and
            (right) used in \cref{algo:ham_changepoint} for Hamiltonian change detection.
        The right-hand plot shows the median and $95\%$ interpercentile range across $120$ trials.
    }
    \label{fig:cusum_arl}
\end{figure}

The large average run lengths shown in \cref{fig:cusum_arl} are intractable to compute by direct simulation, but are instead possible to obtain thanks to specialized SPRT analysis in \cite[Appendix~A.4]{wald1947sequential} and relation to CUSUM by \cite[Eq.~(9)]{page1954continuous}.
This concerns a special case where all possible values of the per-round score $z_i$ are integer multiples of some constant.
Specifically, in the above plots we use $\frac{dQ}{dP}(\text{Reject}) = -2 \frac{dQ}{dP}(\text{Accept})$.
Denoting by $p$ and $q$ the rejection probabilities in the distributions $P = \Ber(p)$ and $Q = \Ber(q)$,
we can obtain this by, for instance, choosing $q = 0.5$ and $p = 1 - \varphi/2 \approx 0.19$ with $\varphi = \frac{1+\sqrt{5}}{2}$ being the golden ratio.
This entails $\frac{dQ}{dP}(\text{Reject}) = \ln\bigl(\frac{q}{p}\bigr) = 2 \ln(\varphi)$ and $\frac{dQ}{dP}(\text{Accept}) = \ln\bigl(\frac{1-q}{1-p}\bigr) = -\ln(\varphi)$, and subsequent constants (i.e., $h$) are rescaled to hide the $\ln(\varphi)$ multiple.
The left-hand plot uses rejection thresholds $h$ from $2$ to $52$ in increments of $5$, and the right-hand plot uses $h = 3$.

\section*{Acknowledgments}
STF, JS, AZ were supported in part by the Commonwealth Cyber Initiative, an investment in the advancement of cyber R\&D, innovation, and workforce development.
M.M acknowledges support from the Institute for Quantum Information and Matter, an NSF Physics Frontiers Center.
STF was supported in part by the U.S.\ Department of Energy, Office of Science, National Quantum Information Science Research Centers, Co-design Center for Quantum Advantage under contract number DE-SC0012704.

\appendix
\phantomsection
\addcontentsline{toc}{section}{References}
\bibliographystyle{alpha-links}
\bibliography{Arxiv_refs}

\section{Robustness of CUSUM}\label{appx:cusum-robust}

The goal of this section is to prove a (perhaps intuitive) result saying that using ``worse'' distributions can only shorten the termination time, and using ``better'' ones only lengthens it.
Similar results (with corresponding lower bounds) have previously appeared in literature for unknown distributions belonging to known classes \cite[Theorem~III.2]{unnikrishnan2011minimax} \cite[Theorem~3]{molloy2017misspecified} and moreover for ones that may be different at each time step \cite[Section~III]{liang2022quickest} \cite[Theorem~1]{molloy2019minimax}---the latter of these matching the below statement.

\begin{lemma}[Following {\cite[Theorem~1]{molloy2019minimax}}]\label{lem:cusum-robust}
    Consider applying CUSUM in two situations with random variables $X_1, \ldots, X_{i-1}, X_{i+1}, \ldots$ following the same (but not necessarily identical) distributions, but $X_i$ following $u$ in one case and $u^\prime$ in the other.
    Denoting the average run lengths of the two situations as $L$ and $L'$, we have
    \begin{subequations}\label{eq:cdf-ineq}
    \begin{align}
        L^\prime &\leq L & \text{when}\quad u^\prime &\succeq u,\label{eq:cdf-ineq-a}\\
        L^\prime &\geq L & \text{when}\quad u^\prime &\preceq u,\label{eq:cdf-ineq-b}
    \end{align}
    \end{subequations}
    with stochastic ordering as defined in \eqref{eq:stochastic-order}.
\end{lemma}
While we do not explicitly require $u'$ to be absolutely continuous with respect to $u$ (i.e., $u' \ll u$), this assumption is implicit.
In particular, the condition of \cref{eq:cdf-ineq-b} prohibits $u' \not\ll u$, while in \cref{eq:cdf-ineq-a} it produces an infinite score that makes CUSUM terminate immediately.

By induction (and limiting behavior over a closed condition), the same average run length (ARL) inequalities hold when each of the distributions differs, so long as the right-hand conditions in \eqref{eq:cdf-ineq} are satisfied for each changed distribution.

Unlike SPRT, the average run length of CUSUM on Bernoulli distributions $\Ber(\theta)$ is monotonic with respect to $\theta$.
This makes sense on an intuitive level---a procedure that terminates only when getting defects frequently enough runs longer when defects are less frequent and shorter when more frequent.
This behavior also allows us to work with distributions that are not identical (but still independent); having individual $\theta_t \geq \theta$ decreases $L$, and $\theta_t \leq \theta$ increases $L$.
We follow a similar line of reasoning for general distributions.

We begin with an observation: giving CUSUM a head start by increasing the initial cumulative score $S_0$ can only shorten its expected runtime.

\begin{lemma}\label{lem:monotonic-z}
    Suppose the distributions of scores $u_i$ are independent and fixed.
    Let $l(x)$ denote the average run length of CUSUM when starting with the initial cumulative score $S_0 = x$.
    We have
    \begin{align}
        l(y) &\leq l(x) & \text{when}\quad y &\geq x.
    \end{align}
\end{lemma}

\begin{proof}
    Fix a sampled sequence of observations $x_1, x_2, \ldots$\ .
    By induction, for any $i$ we have $s_i(y) \geq s_i(x)$.
    As such, this sequence can only terminate earlier when starting with $y$ as opposed to with $x$.
    The sequences are independent of the initial starting score, so the relation holds for ARL $l$ by averaging over such sequences.
\end{proof}

This property lets us focus on analyzing the effect of changing the distribution of scores at a single timestep, propagating the changes to the subsequent step and evaluating the ARL using the above property.

\begin{proof}[Proof of \cref{lem:cusum-robust}]
    As $X_1$ through $X_{i-1}$ follow identical distributions in the two scenarios, the cumulative scores $S_{i-1}$ at the end of step $i-1$ have the same distributions in both scenarios.
    Since the distribution of $X_i$ on the other hand differs, consider the cumulative score at step $i$ is given by $S_i = \pos{S_{i-1} + Z_i}$.
    Denoting cumulative distribution functions by $F$, we begin with an observation that stochastic ordering as defined in \cref{eq:stochastic-order} causes $F_{Z_i}$ to either decrease or increase pointwise.
    Specifically, $u^\prime \succeq u$ is equivalent to $\Pr[ Z_i > x \mid X_i \sim u^\prime] \geq \Pr[ Z_i > x \mid X_i \sim u]$ and hence $1 - F_{Z_i \mid X_i \sim u^\prime}(x) \geq 1 - F_{Z_i \mid X_i \sim u}(x)$; and analogously for $u^\prime \preceq u$.
    We see that $F_{S_i}(s) = \E[F_{Z_i}(s - S_{i-1})]$ for $s \geq 0$ (and $0$ otherwise) due to $S_{i-1}$ and $Z_i$ being independent, meaning that these directional changes in $F_{Z_i}$ are reflected in $F_{S_i}$.
    Although not necessary for the proof, we remark that when $S_{i-1}$ is discrete or continuous, the same conclusion is revealed by $F_{S_i} = (f_{S_{i-1}} \ast F_{Z_i}) \mathbf{1}_{\{s \geq 0\}}$ (where $\mathbf{1}$ is the indicator function) being monotone in $F_{Z_i}$.

    It remains to relate such changes to the average run length $L$.
    Let $l(s)$ be defined as in \cref{lem:monotonic-z} with distributions following the end of step $i$.
    Denoting by $\mu_{S_i}$ the probability measure of $S_i$, we have
    \begin{equation}
        L
        = \int_{[0,\infty)} l(s)\, \mu_{S_i}(ds)
        = \int_0^\infty \mu_{S_i} (\{ s \geq 0 : l(s) \geq x \}) dx
    \end{equation}
    Let $A_x \coloneqq \{ s \geq 0 : l(s) \geq x \}$; by \cref{lem:monotonic-z} each $A_x$ is an interval with $0$ as the left endpoint.
    We then have that $F_{S_i \mid X_i \sim u'} \leq F_{S_i \mid X_i \sim u}$ pointwise implies $\mu_{S_i \mid X_i \sim u'}(A_x) \leq \mu_{S_i \mid X_i \sim u}(A_x)$ and hence $L' \leq L$.
    Analogous reasoning leads to $L' \geq L$ in the other case.

    In particular, let us define $l^\gets(x) \coloneqq \sup A_x$ (analogous to an inverse of $l$).
    By \Cref{lem:monotonic-z}, $A_x$ is $[0, l^\gets(x)]$ (or $[0, l^\gets(x))$ if the supremum is not attained), and we have
    \begin{equation}
        L = \int_0^\infty F_{S_i} (l^\gets(x)) \, dx,
    \end{equation}
    or the same with $l^{\gets}(x)^-$ in the case of an open endpoint.
    Hence, pointwise decreasing or increasing $F_{Z_i}$ leads to corresponding changes in $F_{S_i}$ and $L$.
\end{proof}

\section{Classical certification complexity}\label{section:classical_certification_complexity}
Here, we estimate the classical computation complexity of measurement in the DT for the \cref{algo:state_cert_sub} in the context of Hamiltonian certification.
The hypothesis state is
\begin{equation}\label{eq:hyp_ham}
    \ket*{\tilde\phi_0}
    = \sum_{j=0}^\mathbf{k} \frac{(-i H_0 t)^j}{j !} \ket{\phi},
\end{equation}
where $\ket{\phi}$ is a stabilizer product state.
Suppose the target $n$-qubit Hamiltonian $H_0$ is $\mathbf{m}$-sparse, meaning that there are at most $\mathbf{m}$ nonzero Pauli terms.
Then, there are $\Gamma \coloneqq 1+\mathbf{m}+\mathbf{m}^2 + \dots + \mathbf{m}^\mathbf{k} = O(\mathbf{m}^\mathbf{k})$ terms on the right-hand side of \cref{eq:hyp_ham}.
So, we can write
\begin{equation}\label{eq:hyp}
    \ket*{\tilde\phi_0}
    = \sum_{\gamma=1}^\Gamma \ket{\phi_\gamma}
    = \sum_{\gamma=1}^\Gamma \bigotimes_{i=1}^n \ket*{\phi_\gamma^{(i)}}.
\end{equation}
Note that single-qubit states $\{\ket*{\phi_\gamma^{(i)}}\}$ are not normalized.

As shown in Ref.~\cite[Sec.~2.4]{gupta2025single}, constructing a DT basis involves evaluating expectation values of tensor-product observables.
Therefore, we start with the following statement.
\begin{proposition}\label{prop:exp_value}
    Let $\Omega = \bigotimes_{i \in A} \omega^{(i)} \otimes I^{[n] \setminus A}$ be a $\abs{A}$-local tensor-product operator.
    Calculating $\bra*{\tilde\phi_0} \Omega \ket*{\tilde\phi_0}$ requires $O(\Gamma^2 n)$ operations.
\end{proposition}

\begin{proof}
    Expanding the expectation value produces
    \begin{subequations}
    \begin{align}
        \bra*{\tilde\phi_0} \Omega \ket*{\tilde\phi_0}
        &= \sum_{\gamma=1}^\Gamma \sum_{\gamma'=1}^\Gamma
        \bra{\phi_\gamma} \Omega \ket{\phi_{\gamma'}} \\
        &= \sum_{\gamma=1}^\Gamma \sum_{\gamma'=1}^\Gamma
        \Bigg[
        \prod_{a \in A}
        \bra*{\phi_\gamma^{(a)}} \omega^{(a)} \ket*{\phi_{\gamma'}^{(a)}}
        \prod_{b \notin A}
        \braket*{\phi_\gamma^{(b)}}{\phi_{\gamma'}^{(b)}}
        \Bigg].
    \end{align}
    \end{subequations}
    The expression above involves estimating $\Gamma^2 n$ inner products.
    Using $\bra*{\phi_\gamma^{(j)}}\omega^{(j)}\ket*{\phi_{\gamma'}^{(j)}}^*
    = \bra*{\phi_{\gamma'}^{(j)}}\omega^{(j)}\ket*{\phi_\gamma^{(j)}}$,
    we can compute each unordered pair $(\gamma,\gamma')$ once, giving
    $\binom{\Gamma+1}{2} n$ operations, which is still $O(\Gamma^2 n)$.
\end{proof}

Here, we show two ways to construct the DT basis required for the \textsc{Certify} algorithm.
We start with a naive approach and then slightly modify it to obtain better runtime.

\begin{theorem}[Building DT basis]
    Building DT basis for the hypothesis state in \cref{eq:hyp} takes $O(\Gamma^2 n^2)$ steps.
\end{theorem}
\begin{proof}
As detailed in Ref.~\cite[Sec.~2.4]{gupta2025single}, to measure in the DT basis, we are required to calculate 6 expectation values of tensor-product operators $n-k$ times.
Combining this with \cref{prop:exp_value}, we obtain scaling $\bigO(\Gamma^2n(n-k))$.
Since $\E[k] = n/2$, on average we would need to calculate $\bigO(\Gamma^2 n^2)$ inner products to get a DT basis for a single run of \cref{algo:state_cert_sub}.
\end{proof}

Thanks to the recursive nature of DT construction, many calculations used to obtain $\ket{b_i}$ can be reused to obtain $\ket{b_{i+1}}$, resulting in

\begin{theorem}[Building DT basis; better approach]
    Building DT basis for the hypothesis state in \cref{eq:hyp} takes $O(\Gamma^2 n)$ steps.
\end{theorem}
\begin{proof}
    At iteration $t \in [n-k]$ of DT basis construction, we need to evaluate
    \begin{equation}
         n_{t,g}^q \coloneqq \bra*{\tilde\phi_0} (\kb{x} \otimes \kb{q}^{(k)} \otimes \kb{b_1}^{(k+1)} \otimes \dots \otimes \kb{b_{t-1}}^{(k+t-1)} \otimes \kb{g}^{(k+t)}) \ket*{\tilde\phi_0},
    \end{equation}
    where $x \in \{0,1\}^{k-1}$ is a bitstring sampled in step 2 of \cref{algo:state_cert_sub}, $\{b_1, \dots, b_{t-1}\}$ are computed in previous iterations, and $\ket*{\tilde\phi_0}$ is as in \cref{eq:hyp_ham}.
    We need to do this for each choice of $q \in \{0,1\}$ and $g \in \{0,+,i\}$.
    
    We expand the expectation value above using the definition of $\ket*{\tilde\phi_0}$:
    \begin{subequations}
    \begin{align}
        n_{t,g}^q &= \sum_{\gamma, \gamma'=1}^{\Gamma} \bra{\phi_\gamma}
        (\kb{x} \otimes \kb{q}^{(k)} \otimes \kb{b_1}^{(k+1)} \otimes \dots \otimes \kb{b_{t-1}}^{(k+t-1)} \otimes \kb{g}^{(k+t)}) \ket{\phi_{\gamma'}} \\
        &= \sum_{\gamma, \gamma'=1}^{\Gamma}
        \prod_{i=1}^{k-1} \braket*{\phi_\gamma^{(i)}}{x^{(i)}}\braket*{x^{(i)}}{\phi_{\gamma'}^{(i)}} \times \braket*{\phi_\gamma^{(k)}}{q}\braket*{q}{\phi_{\gamma'}^{(k)}} \times \prod_{j=k+1}^{k+t-1} \braket*{\phi_\gamma^{(j)}}{b_{j-k}}\braket*{b_{j-k}}{\phi_{\gamma'}^{(j)}} \\
        &\hspace{5em}\times \braket*{\phi_\gamma^{(k+t)}}{g} \braket*{g}{\phi_{\gamma'}^{(k+t)}} \times \prod_{l=k+t+1}^{n} \braket*{\phi_\gamma^{(l)}}{\phi_{\gamma'}^{(l)}}.
    \end{align}
    \end{subequations}
    Four out of five factors above do not include $\ket{b_i}$ and thus can be pre-computed, i.e.~estimated at $t=0$.
    In particular,
    \begin{itemize}
        \item $\prod_{i=1}^{k-1} \braket*{\phi_\gamma^{(i)}}{x^{(i)}}\braket*{x^{(i)}}{\phi_{\gamma'}^{(i)}}$ requires evaluating $\braket*{\phi_\gamma^{(i)}}{x^{(i)}} = (\braket*{x^{(i)}}{\phi_{\gamma}^{(i)}})^*$ for $i \in [k-1]$ and $\gamma \in [\Gamma]$, amounting to $O(\Gamma k)$ inner products;
        \item $\braket*{\phi_\gamma^{(k)}}{q}\braket*{q}{\phi_{\gamma'}^{(k)}}$ requires evaluating $\braket*{\phi_\gamma^{(k)}}{q} = (\braket*{q}{\phi_\gamma^{(k)}})^*$ for $\gamma \in [\Gamma]$ and $q\in \{0,1\}$, amounting to $O(\Gamma)$ inner products;
        \item $\braket*{\phi_\gamma^{(k+t)}}{g} \braket*{g}{\phi_{\gamma'}^{(k+t)}}$ requires evaluating $\braket*{\phi_\gamma^{(k+t)}}{g} \braket*{g}{\phi_{\gamma'}^{(k+t)}}$ for $\gamma, \gamma' \in [\Gamma]$ and $t\in [n-k]$, amounting to $O(\Gamma^2(n-k))$ inner products;
        \item $\prod_{l=k+t+1}^{n} \braket*{\phi_\gamma^{(l)}}{\phi_{\gamma'}^{(l)}}$ requires evaluating $\braket*{\phi_\gamma^{(l)}}{\phi_{\gamma'}^{(l)}}$ for $\gamma, \gamma' \in [\Gamma]$ and $l = k+2 \dots n$, amounting to $O(\Gamma^2(n-k))$ inner products.
    \end{itemize}
    Therefore, the pre-computing stage can be completed in $O(\Gamma^2(n-k))$ time.
    Finally,
    \begin{itemize}
        \item $\prod_{j=k+1}^{k+t-1} \braket*{\phi_\gamma^{(j)}}{b_{j-k}}\braket*{b_{j-k}}{\phi_{\gamma'}^{(j)}}$ requires evaluating $\braket*{\phi_\gamma^{(j)}}{b_{j-k}} = (\braket*{b_{j-k}}{\phi_{\gamma}^{(j)}})^*$ for $\gamma \in [\Gamma]$ and $j=k+1 \dots k+t-1$.
        However, the overlaps have already been calculated for $j\le k+t-2$ during previous steps, meaning that at step $t$, we only need to estimate $\braket*{\phi_\gamma^{(j)}}{b_{t-1}}$ for $\gamma \in [\Gamma]$, amounting to $O(\Gamma)$.
        Repeating this for $t \in [n-k]$ gives complexity $O(\Gamma(n-k))$.
    \end{itemize}
    
    In summary, the runtime of constructing the DT basis is dominated by the pre-computing stage which requires evaluating $O(\Gamma^2(n-k))$ inner products.
    Substituting $\E[k] = n/2$, we obtain $O(\Gamma^2 n)$.
\end{proof}

\end{document}